\documentclass{article}

\usepackage{amsmath}
\usepackage{amsthm}

\usepackage{newtxtext}
\usepackage[subscriptcorrection]{newtxmath}

\usepackage{natbib}

\usepackage{mathtools}
\usepackage{tabularx}
\usepackage{boondox-calo}
\usepackage{nicematrix}
\usepackage{tikz}
\usetikzlibrary{positioning, arrows.meta, calc}
\usepackage{xcolor}
\usepackage{blkarray}
\usepackage{enumitem}
\usepackage{booktabs}  
\usepackage{multirow}  
\usepackage{graphicx}
\usepackage{caption}

\graphicspath{{./art/}}

\usepackage[plain,noend]{algorithm2e}

\makeatletter
\renewcommand{\algocf@captiontext}[2]{#1\algocf@typo. \AlCapFnt{}#2} % text of caption
\def\@algocf@capt@plain{top}
\renewcommand{\algocf@makecaption}[2]{%
	\addtolength{\hsize}{\algomargin}%
	\sbox\@tempboxa{\algocf@captiontext{#1}{#2}}%
	\ifdim\wd\@tempboxa >\hsize%     % if caption is longer than a line
	\hskip .5\algomargin%
	\parbox[t]{\hsize}{\algocf@captiontext{#1}{#2}}% then caption is not centered
	\else%
	\global\@minipagefalse%
	\hbox to\hsize{\box\@tempboxa}% else caption is centered
	\fi%
	\addtolength{\hsize}{-\algomargin}%
}
\makeatother

\def\cl{\mathcal{l}}
\def\cL{\mathcal{L}}
\DeclareMathOperator{\Tr}{Trace}

\usepackage{arxiv}

\usepackage[utf8]{inputenc} % allow utf-8 input
\usepackage[T1]{fontenc}    % use 8-bit T1 fonts
\usepackage{url}            % simple URL typesetting
\usepackage{amsfonts}       % blackboard math symbols
\usepackage{nicefrac}       % compact symbols for 1/2, etc.
\usepackage{microtype}      % microtypography
\usepackage{lipsum}
\graphicspath{ {./images/} }

\newtheorem{theorem}{Theorem}
\newtheorem{definition}{Definition}
\newtheorem{lemma}{Lemma}
\newtheorem{remark}{Remark}
\newtheorem{assumption}{Assumption}
\newtheorem{corollary}{Corollary}
\usepackage{hyperref}
\hypersetup{colorlinks=true}

\title{Hierarchical Clustering of Networks via Hierarchical Distance Matrices}

\author{
	Li Chen \\
	School of Mathematics, Southwest Minzu University, \\ 168 Wenxing Section of Dajian Road, Shuangliu District, Chengdu 610225, China
	\texttt{lchen@swun.edu.cn}\\	
	\And
	Nathaniel Josephs \\
	Department of Statistics, North Carolina State University, \\Raleigh, NC 27695, USA
	\texttt{nathaniel.josephs@ncsu.edu}\\
	\And
	Eric D. Kolaczyk \\
	Department of Mathematics and Statistics, McGill University, \\ 805 Sherbrooke Street West, Montr\'{e}al, QC H3A 0G4, Canada. \texttt{eric.kolaczyk@mcgill.ca}\\
	\And
	Lizhen Lin\\
	Department of Mathematics, The University of Maryland, \\ College Park, MD 20742, USA
	\texttt{lizhen01@umd.edu}
}

\begin{document}	
	
	\maketitle
	
	\begin{abstract}
		Clustering populations of networks while recovering their latent hierarchical organization is a fundamental yet largely unexplored problem in network analysis. 
		To formalize this, we introduce the Hierarchical Distance Matrix, a specific class of population-level distance matrices that encodes latent hierarchical organization through recursively nested distance separation, accommodating unbalanced tree depths.
		Building on this framework, we propose a fully data-driven top-down procedure: network hierarchical clustering based on two-sample testing (NHC-TST). The algorithm recursively splits networks via spectral clustering and uses a graph-based two-sample stopping rule. The procedure adaptively determines the branching structure without requiring prior knowledge of the number of clusters or tree depth.
		Theoretically, we establish exact recovery of the population-level hierarchical structure and statistical consistency in the empirical procedure.
		Simulation studies demonstrate highly accurate recovery of both cluster memberships and hierarchical relationships across a wide range of settings. Applied to a global migration dataset, NHC-TST uncovers interpretable multi-resolution temporal structures that are not revealed by conventional flat clustering approaches.
	\end{abstract}
	

	\section{Introduction}
	
	Multiple network datasets arise increasingly often across a broad range of scientific domains such as in biology (genomics networks \citep{sole2002complex}, brain connectomics \citep{bassett2008hierarchical}), engineering (computer networks \citep{leskovec2007graph}, transportation networks \citep{kurant2006extraction}), and the social sciences (social networks \citep{eagle2009inferring}, organizational networks \citep{josephs2024communication}).
	Often, there are no \textit{response} labels associated with each network so subsequent analyses must be unsupervised. A classic unsupervised learning task is to cluster the data into meaningful groups.
	Traditional clustering is flat, meaning that there is no structure between the inferred groups.
	However, analysts are often interested in the (dis)similarity between clusters, which explains the popularity of hierarchical clustering techniques. Hierarchical clustering provides links between clusters, typically using a bottom-up agglomerative or top-down divisive strategy, and clusters are considered more similar if they are closer with respect to the shortest path on the resulting hierarchy. 
	
	There are several existing approaches to clustering networks.
	\cite{josephs2023nested} provides a summary of these methods, categorizing them into three general approaches.
	The first approach is the random-effects model, in which network clusters are subject to Markov perturbations.
	These methods include \cite{paul2020random, chen2022global}.
	The second approach is the measurement-error model, in which the network itself is subject to Markov perturbations.
	Clustering methods following this approach include \cite{le2018estimating, mantziou2021bayesian, young2022clustering, josephs2021network}.
	The third approach considers each network as a separate layer of a multilayer network and includes \cite{stanley2016clustering, jing2021community, fan2022alma, signorelli2020model}.
	There are also a few clustering methods that fall outside of these categories.
	\cite{mukherjee2017clustering} provide two graph-clustering algorithms based on spectral clustering of the pairwise distance matrix between the estimated graphons of these networks: Network Clustering based on Graphon Estimates (NCGE) for vertex-aligned networks using graphon estimates, and Network Clustering based on Log Moments (NCLM) for unaligned networks utilizing spectral moment features.
	\cite{reyes2016stochastic} and \cite{josephs2023nested} provide Bayesian nonparametric approaches that build on the Bayesian stochastic block model (SBM).
	In particular, \citet{josephs2023nested} proposes a nested stochastic block model (NSBM), which employs a nested Dirichlet process to cluster network collections.
	All of these methods for clustering networks are flat. 
	
	Currently, there are only two approaches to hierarchical clustering specifically designed for networks.
	\cite{diquigiovanni2019analysis} use the Louvain method for community detection on each network and then
	perform agglomerative clustering of the networks based on the Rand index between their community structures. 
	%\cite{rebafka2024model} similarly use an agglomerative approach based on the integrated classification likelihood.
	\citet{rebafka2024model} proposes a Bayesian framework for mixtures of SBMs, employing a hierarchical agglomerative algorithm based on the integrated classification likelihood (ICL), which we refer to as hierarchical ICL for SBM mixtures (HICL-SBM).
	Both methods, like traditional hierarchical clustering algorithms, produce a hierarchy of the clusters. However, this \textit{dendrogram} is simply a trace of greedy splits towards a flat clustering rather than recovering a latent hierarchical organization.
	Currently, there are no methods to recover a hierarchical structure that is meaningfully related to cluster splits.
	
	To address this gap, we propose a novel approach that directly recovers both the network clusters and their latent hierarchical relationships.
	To make this problem statistically well defined, we first formulate what it means for a distance (dissimilarity) matrix to exhibit a hierarchical structure by introducing a recursive separation condition, formalized as the Hierarchical Distance Matrix. This characterization captures the multiscale nature of hierarchical clustering by requiring clusters that separate earlier in the tree to exhibit larger population dissimilarities than those separated at deeper levels, while naturally allowing branches to terminate at different depths. 		
	For network populations sampled under this condition, we propose a top-down algorithm to recover the latent tree and cluster memberships. Our procedure recursively partitions the networks via spectral clustering, employing a graph-based two-sample test as an adaptive stopping rule.	
	On the theoretical side,  we prove the consistency of our method when the distance matrix is known, as well as the consistency of our method using a plug-in estimator for the distance matrix. To evaluate our method on synthetic data, we propose a novel adaptation of the Cophenetic Correlation Coefficient for hierarchical structures.
	
	Herein, we are motivated by global migration flows which impact economic, social, and political landscapes, driving demographic shifts and influencing policy frameworks worldwide \citep{right2016assessing}.
	% Quantifying migration is crucial for policymakers, humanitarian organizations, and researchers addressing issues ranging from labor markets to emergency response. 
	% Traditional migration data collection methods rely on national censuses and administrative records, which suffer from missingness, inconsistency across nations, and delayed reporting \citep{un2015international, desa2023international}.
	% Recently, digital trace data has emerged as an alternative source of migration data without the same limitations \citep{sirbu2021human, zagheni2017leveraging, hawelka2014geo, lu2016unveiling, wesolowski2015quantifying, zagheni2014inferring}.
	\citet{chi2025measuring} utilize a novel dataset constructed from privacy-protected Facebook user data to estimate monthly country-to-country migration flows from January 2019 to December 2022.
	This high-resolution dataset encompasses periods of significant global disruption, such as the COVID-19 pandemic and geopolitical crises, including the invasion of Ukraine. These global bilateral migration dynamics can be naturally represented as a sequence of temporal networks. Each month in the dataset is represented by an individual network, where nodes denote countries and edges reflect migration volumes.
	% Because the countries remain constant across the time period, the resulting dataset is a collection of labeled networks, i.e. there is a known node correspondence between the networks.
	In the context of global migration, we may want to know not only if COVID-19 was disruptive, but also whether migration returned to pre-COVID behavior or was permanently altered. This degree of difference is relevant to policymakers and cannot be determined from flat clusters alone.

	\subsection{Notation}
	%\textbf{Notation.}
	The following notation is used throughout this paper. Let $\mathbb{I}(\cdot)$ denote the indicator function. We use $f(n)=\omega(g(n))$ to mean $\lim_{n\to\infty} f(n)/g(n)=\infty$. Let $\mathbf{1}$ denote the all-ones vector of appropriate dimension. For a matrix $M$, let $M_{ij}$ denote its $(i,j)$-th entry and $\|M\|_F$ its Frobenius norm. When $M \in \mathbb{R}^{r \times r}$, denote its ordered eigenvalues by $\lambda_1(M) \ge \cdots \ge \lambda_r(M)$. For conciseness, we refer to the eigenvectors associated with the largest and smallest eigenvalues as the \emph{dominant} and \emph{minimal} eigenvectors, respectively. For a matrix $A \in \mathbb{R}^{n \times n}$,  we denote its $i$th row sum by $\deg_i(A) = \sum_{j=1}^{n} A_{ij}$ and define the combinatorial Laplacian as  $L_A = \mathrm{diag}(A\mathbf{1}) - A$.
	
	\section{Hierarchical clustering of networks}
	\label{sec:model}
	
	\subsection{Model and problem formulation}
	
	Let $A^{(1)}, \ldots, A^{(m)}$ denote $m$ observed network objects defined on a common set of $n$ nodes, where each $A^{(r)} \in \{0,1\}^{n \times n}$, $r=1,\ldots, m$, is the adjacency matrix of an undirected simple graph on $n$ nodes. That is, for $i<j$, $A^{(r)}_{ij}=1$ if nodes $i$ and $j$ are connected and $A^{(r)}_{ij}=0$ otherwise, with $A^{(r)}_{ji}=A^{(r)}_{ij}$ and $A^{(r)}_{ii}=0$.
	We assume that these networks arise from $K$ latent populations, each characterized by a link-probability matrix $\mathcal P^{(k)} \in [0,1]^{n \times n}$, $k=1,\ldots,K$. Specifically, for each network $r$, there exists a latent label $k_r \in \{1,\ldots,K\}$ such that
	\begin{equation*}
		A^{(r)}_{ij} \mid \mathcal P^{(k_r)}_{ij} \sim \mathrm{Bernoulli}(\mathcal P^{(k_r)}_{ij}), \qquad 1 \le i < j \le n,
	\end{equation*}
	independently across pairs $(i,j)$.
	
	A key feature of our model is that the $K$ populations are organized according to an unknown hierarchical tree of depth $\mathcal{L}$ (the depth of root is 0). Each cluster $k$ is associated with a binary label
	\[
	x_k = a_1 a_2 \cdots a_{\ell_k}, \quad a_\ell \in \{0,1\}, \ \cl = 1, \ldots, \cl_k, \ \ell_k \le \mathcal{L},
	\]
	which encodes its position in the tree. The sequence $x_k$ records the path from the root to cluster $k$, allowing for potentially unbalanced hierarchical structures. Figure~\ref{fig:schematic} illustrates a toy example of such a tree. At any splitting step $\cl$, we define the binary indicator $a_{\cl} = 0$ if the subset is partitioned into the left branch, and $a_{\cl} = 1$ if it is routed to the right. As shown, the subscript prefixes of each node in the tree naturally matches its binary labels, %As shown, the subscript of each node in the tree naturally matches its binary labels, 
	tracing the recursive branching process down to the terminal leaves. 
	
	% \begin{figure}[!htb]
		% \centering
		% \includegraphics[width=.7\linewidth]{Figs/schematic.png}
		% \caption{Schematic of hierarchical tree for networks. In this toy example, each parent node corresponds to a network that splits into two children nodes based on the addition (red solid) or deletion (blue dashed) of edges. The networks that are observed in a sample are noisy realizations sampled conditional on the networks in the five terminal nodes $(G_{00}, G_{010}, G_{011}, G_{10}, \ \text{and} \ G_{11})$}.
		% \label{fig:schematic}
		% \end{figure}
	
	\begin{figure}[htbp]
		\centering
		\resizebox{\textwidth}{!}{ 
			\begin{tikzpicture}[
				every node/.style={align=center},
				% Slightly adjusted font size to \LARGE to ensure spacing room
				prob_node/.style={draw, line width=1.5pt, rounded corners, fill=blue!10, minimum width=3.8cm, minimum height=1.8cm, font=\Huge},
				leaf_node/.style={draw, dashed, line width=1.5pt, rounded corners, fill=green!10, minimum width=4.5cm, minimum height=1.8cm, font=\Huge},
				arrow/.style={->, >=Stealth, line width=2pt}
				]
				
				% Root
				%\node[prob_node] (Root) at (0, 0) {$\mathcal{P}_{\text{root}}$};
				\node[prob_node] (Root) at (0, 0) {\{$\mathcal{P}_{000}, \mathcal{P}_{001}, \mathcal{P}_{01}, \mathcal{P}_{10}, \mathcal{P}_{11}$\}};
				
				% Level 1
				%\node[prob_node] (G0) at (-10, -3.5) {$\mathcal{P}_{0}$};
				%\node[prob_node] (G1) at (10, -3.5) {$\mathcal{P}_{1}$};
				\node[prob_node] (G0) at (-10, -3.5) {$0$: \{$\mathcal{P}_{000}, \mathcal{P}_{001}, \mathcal{P}_{01}$\}};
				\node[prob_node] (G1) at (10, -3.5) {$1$: \{$\mathcal{P}_{10}, \mathcal{P}_{11}$\}};
				
				\draw[arrow] (Root) -- (G0);
				\draw[arrow] (Root) -- (G1);
				
				% Level 2
				%\node[prob_node] (G00) at (-14.5, -7.0) {$\mathcal{P}_{00}$};
				\node[prob_node] (G00) at (-14.5, -7.0) {$00$:  \{$\mathcal{P}_{000}, \mathcal{P}_{001}$\}};
				\node[prob_node] (G01) at (-5.5, -7.0) {$01$: $\mathcal{P}_{01}$};        
				
				\node[prob_node] (G10) at (5.5, -7.0) {$10$: $\mathcal{P}_{10}$};
				\node[prob_node] (G11) at (14.5, -7.0) {$11$: $\mathcal{P}_{11}$};
				
				\draw[arrow] (G0) -- (G00);
				\draw[arrow] (G0) -- (G01);
				\draw[arrow] (G1) -- (G10);
				\draw[arrow] (G1) -- (G11);
				
				% Level 3 (Asymmetric split)
				\node[prob_node] (G000) at (-17.5, -10.5) {$000$: $\mathcal{P}_{000}$};
				\node[prob_node] (G001) at (-11.5, -10.5) {$001$: $\mathcal{P}_{001}$};
				
				\draw[arrow] (G00) -- (G000);
				\draw[arrow] (G00) -- (G001);
				
				% Observed Collections (Leaves)
				\node[leaf_node] (C000) at (-17.5, -14.5) {$\{A^{(i)}\} \sim \mathcal{P}_{000}$};
				\node[leaf_node] (C001) at (-11.5, -14.5) {$\{A^{(i)}\} \sim \mathcal{P}_{001}$};
				\node[leaf_node] (C01) at (-5.5, -14.5) {$\{A^{(i)}\} \sim \mathcal{P}_{01}$};
				\node[leaf_node] (C10) at (5.5, -14.5) {$\{A^{(i)}\} \sim \mathcal{P}_{10}$};
				\node[leaf_node] (C11) at (14.5, -14.5) {$\{A^{(i)}\} \sim \mathcal{P}_{11}$};
				
				% Dashed sampling arrows
				\draw[arrow, dashed] (G000) -- (C000);
				\draw[arrow, dashed] (G001) -- (C001);
				\draw[arrow, dashed] (G01) -- (C01);
				\draw[arrow, dashed] (G10) -- (C10);
				\draw[arrow, dashed] (G11) -- (C11);
				
			\end{tikzpicture}
		}
		\caption{Hierarchical tree of the underlying link-probability matrices from which populations of networks are drawn. The generative structure accommodates unbalanced tree depths.}     
%		\small \textbf{Alt text:} An unbalanced hierarchical binary tree depicting how network populations recursively split into sub-populations, each corresponding to a specific link-probability matrix at the leaves. The left branch reaches depth three, while the remaining branches terminate at depth two.
		\label{fig:schematic}
	\end{figure}
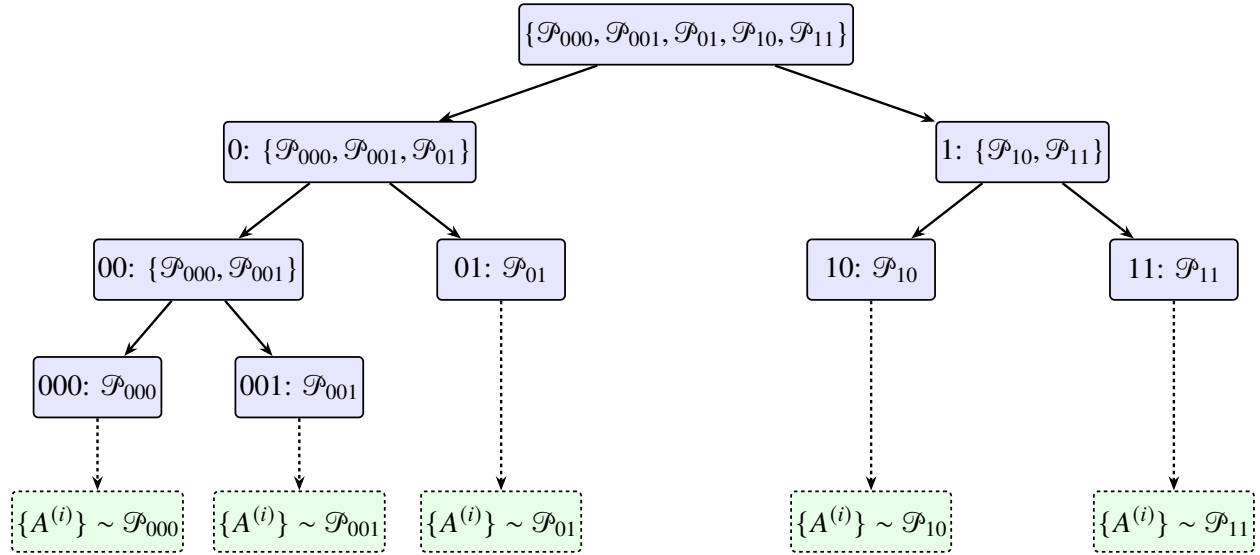
	
	Our goal is to recover both of the following:
	(1) the graph-level cluster assignment $k_r$ for each observed network; 
	and (2) the underlying hierarchical tree structure governing the $K$ populations.
	
	%\subsection{Hierarchical distance matrix structure}
	\subsection{Distance metrics among networks}
	To quantify the structural discrepancies among networks, we establish three distance matrices in this section. We first introduce a tree-based topological distance matrix and a population distance matrix as our foundational metrics. Subsequently, by imposing a specific structural condition on a population distance matrix, we formally define the Hierarchical Distance Matrix (HDM), which serves as a pivotal concept and provides core properties for the theoretical analysis of our clustering algorithm.

	\textbf{Tree distance on binary labels.}
	% To quantify the structural discrepancies among the $K$ clusters, we first define a metric based on hierarchical topology.
	As our first foundational metric, we formalize the tree-based topological distance to capture the structural separation among the $K$ clusters.
	For two \textit{binary labels} $x_k=a_1a_2\cdots a_{\ell_{k}} $ and $x_{k'}=a_1'a_2'\cdots a_{\ell_{k'}}'$, we begin by defining the \textit{earliest splitting level}
	\[
	s(x_k,x_{k'}) \;=\;
	\begin{cases}
		\min\bigl\{\ell \in \{1,\ldots, \mathcal{L}\} : a_\ell \neq a'_\ell \bigr\}, & x_k \neq  x_{k'},\\[4pt]
		\mathcal{L} +1, & x_k= x_{k'}.
	\end{cases}
	\]
	Let $\text{lcp}(x_k,x_{k'})$ denote the length of the \textit{longest common prefix} of the two sequences, then $\text{lcp}(x_k,x_{k'})~=~s(x_k,x_{k'})-1$.
	Using this notion, we define the \textit{tree distance} between two clusters $k$ and $k'$ as
	\begin{equation}\label{eq:tree-dist}
		d_{\mathcal{T}}(x_k,x_{k'})  =  \mathcal{L} -\text{lcp}(x_k,x_{k'}) = \mathcal{L} + 1- s(x_k,x_{k'}).
	\end{equation}
	If $x_k$ and $x_{k'}$ split \emph{earlier}, then $s(x_k,x_{k'})$ is smaller and hence their shared prefix is shorter, corresponding to a larger tree distance $d_{\mathcal{T}}(x_k,x_{k'})$.
	Therefore, a larger tree distance means the clusters are less similar.
	
	For example, as illustrated in Figure~\ref{fig:schematic}, consider the terminal clusters with binary labels $x_1 = 000$ and $x_2 = 01$, respectively. The longest common prefix of these two labels is $\text{lcp}(x_1, x_2) = 1$, implying that these two clusters share the same path down to the first layer and diverge at the second layer of the tree. Consequently, their earliest splitting level is $s(x_1, x_2) = 2$ and the tree distance is $d_{\mathcal{T}}(x_1,x_2) = 2$. 
	
	\textbf{Population distance matrix across observed networks.}
	While the tree distance matrix captures the hierarchical topology, we also quantify structural differences directly through the true link-probability matrices. We define the normalized population distance matrix (or simply the distance matrix) between $m$ networks as $D \in \mathbb{R}^{m \times m}$, where its $(r,s)$-th entry represents the pairwise distance is
	\begin{equation} \label{p-dis}
		D_{rs} = \frac{1}{n} \|\mathcal P^{(k_r)} - \mathcal P^{(k_s)}\|_F.
	\end{equation}
	Let $Z\in\{0,1\}^{m\times K}$ be the membership matrix with $Z_{rk}=1$ iff $k_r=k$.
	Define $\mathcal{D}\in\mathbb{R}^{K\times K}$ by $\mathcal{D}_{kk'}= \frac{1}{n} \|\mathcal P^{(k)} - \mathcal P^{(k')}\|_F$, with $\mathcal{D}_{kk}=0$.
	Then
	\begin{equation*}
		D = Z\,\mathcal{D}\,Z^\top.
	\end{equation*}
	If one instead sets $\mathcal{D}_{kk}=\rho$ for a constant $\rho$, then $D=Z\mathcal{D}Z^\top-\rho ZZ^T$.  
	
	This matrix encodes pairwise dissimilarities among observed networks and serves as the key object for clustering.
	
	\textbf{Hierarchical distance matrix structure.}
	We now introduce a structural condition that links the hierarchical tree to the population distance matrix $D$.
	Without loss of generality, suppose the observed networks are reordered so that networks belonging to the same cluster occupy adjacent rows of $Z$.
	%We consider a class of distance matrices that admit both a hierarchical structure and a monotonic property, a structure naturally applicable to our population distance matrix. 
	To formally characterize this matrix structure, we first introduce a hierarchical indexing scheme based on binary prefix sequences. Assume $D \in \mathbb{R}^{m\times m}$ is any symmetric matrix. % with $D_{rr}=0$ and $D_{rs}\ge 0$ for all $r \neq s$.  
	For a given prefix $x$, let $D^{[x]}$ denote the submatrix of $D$ indexed by the networks whose binary labels begin with $x$.
	Therefore, for any \emph{non-terminal} prefix $x$, i.e., an internal node that splits into branches $x0$ and $x1$, the submatrix $D^{[x]}$ can be partitioned into branch-specific diagonal blocks and a cross-branch off-diagonal block $R^{[x]}$:
	\begin{equation}\label{eq:Dx}
		D^{[x]} =
		\begin{pNiceArray}{c|c}[cell-space-limits=2pt]
			D^{[x0]} & R^{[x]}\\
			\Hline
			(R^{[x]})^\top & D^{[x1]}
		\end{pNiceArray}.
	\end{equation}
	
	%With this structural indexing scheme, we formally define the hierarchical matrix and hierarchical distance matrix as follows.
	
	With this indexing scheme, we establish the following formal definitions.
	
	\begin{definition}[Hierarchical Matrix]\label{def:HM}
		For a symmetric matrix $D \in \mathbb{R}^{m\times m}$ with $D_{rr}=0$ and $D_{rs}\ge 0$  for all $r \neq s$, we say that $D$ is a Hierarchical Matrix (HM) if there exists a permutation of its indices such that for every \emph{non-terminal} prefix $x$, the submatrix $D^{[x]}$ can be recursively partitioned as in \eqref{eq:Dx}, subject to the following conditions:
		\begin{enumerate}[label=(\roman*)]
			\item \emph{Strict monotonicity:} The elements of the off-diagonal block strictly bound the elements of the diagonal blocks, such that
			\begin{equation*}
				\max_{i,j}\left\{D^{[x0]}_{ij},\, D^{[x1]}_{ij}\right\} \;<\; \min_{i,j} \left(R^{[x]}_{ij}\right).
			\end{equation*}
			\item \emph{Recursive hierarchy:} The diagonal blocks $D^{[x0]}$ and $D^{[x1]}$ are themselves HMs of appropriate dimensions or scalar zeros at the terminal leaf level.
		\end{enumerate}
		%\hfill $\square$
	\end{definition}
	
	\begin{definition}[Hierarchical Distance Matrix]\label{def:HDM}
		The normalized population distance matrix $D$ defined in \eqref{p-dis} is called a Hierarchical Distance Matrix (HDM) if it is an HM.
	\end{definition}

	In particular, the permuted HDM $D$ at the first layer admits
	\begin{equation}\label{eq:D}
		D =
		\begin{pNiceArray}{c|c}[cell-space-limits=2pt]
			D^{[0]} & R\\
			\Hline
			R^\top & D^{[1]}
		\end{pNiceArray},
	\end{equation}
	where $D^{[0]}$ and $D^{[1]}$ are within-cluster distance matrices and $R$ contains between-cluster distances. Entry values of $D^{[0]}$ and $D^{[1]}$ are smaller than those of $R$.
	The recursive hierarchy condition means that this block-partitioning scheme can be iterated down the tree.
	%Furthermore, the entrywise inequalities implied by Definition~\ref{def:HDM} yield the monotonic separation across layers (e.g., within $D^{[x0]}$ or $D^{[x1]}$ is smaller than across $R^{[x]}$, and distances across the child off-diagonal blocks $R^{[x0]}$ and $R^{[x1]}$ are smaller than those across their parent block $R^{[x]}$). %and distances across $R^{[x]}$ are smaller than those across higher-level off-diagonal blocks).
	We illustrate the structure of the HDM and show how its underlying assumptions are satisfied in the simulation setting (see Section~\ref{sec:sim}). The first two simulation scenarios provide explicit examples.
	
	\begin{remark}
		With a slight abuse of notation, we denote both HM and HDM by $D$, which allows us to directly apply existing HM expressions in subsequent analyses.
	\end{remark}
	
	\begin{remark}
		When $D$ is an HDM, the strict monotonicity condition intuitively guarantees consistency with the tree distance $d_{\mathcal{T}}$: clusters diverging earlier exhibit larger distances. 
		Specifically, consider three classes $k, k', k''$ such that $d_{\mathcal{T}}(k,k') > d_{\mathcal{T}}(k,k'')$, meaning $k$ and $k'$ split at an earlier layer $i$ than $k$ and $k''$. For any networks $r, r', r''$ in clusters $k, k', k''$ respectively, their pairwise distance $D_{rr'}$ falls into an off-diagonal block at layer $i$, whereas $D_{rr''}$ remains in a diagonal block. By the strict monotonicity of the HDM, we immediately have $D_{rr'} > D_{rr''}$, ensuring that the true hierarchical topology is fully preserved in $D$.
	\end{remark}
	
	%\begin{remark}	[Distance Consistency]
	%When the population distance matrix $D$ is an HDM, the strict monotonicity condition in Definition \ref{def:HDM} intuitively guarantees  that $D$ is consistent with the underlying tree distance $d_{\mathcal{T}}$: 
	%clusters diverging earlier in the tree exhibit larger distances in $D$.
	%Specifically, consider three distinct classes $k, k'$ and $k''$ with hierarchical binary labels $x_k, x_{k'}$ and $x_{k''}$. Assume $\boldsymbol{d_{\mathcal{T}}(k,k') > d_{\mathcal{T}}(k,k'')}$, which implies that classes $k$ and $k'$ split at an earlier layer $i$, whereas $k$ and $k''$ split at a deeper layer $j$ ($i < j $). Let $a_1 \cdots a_{i-1}$ denote their shared prefix at layer $i-1$, and let $r, r', r''$ be the indices of the networks belonging to clusters $k, k', k''$, respectively. Because $k$ and $k'$ diverge at layer $i$, their corresponding distance $D_{rr'} = Z_{r \cdot} \, \mathcal{D}\,(Z_{r' \cdot})^\top$ is in the off-diagonal block  of the submatrix $D^{[a_1 \cdots a_{i-1}]}$. On the other hand, since $k$ and $k''$ share a longer prefix,% up to layer $j-1$, 
	%they belong to the same cluster at layer $i$. Thus, their distance $D_{rr''}$ is in a diagonal block of $D^{[a_1 \cdots a_{i-1}]}$. By the strict monotonicity condition of the HDM, we immediately have $\boldsymbol{D_{rr'} > D_{rr''}}$. This property ensures that the true hierarchical relationships are fully preserved within the population distance matrix.
	%\end{remark}
	
	\subsection{Hierarchical clustering algorithm}
	In this subsection,  we propose a recursive algorithm for clustering the observed networks and recovering their hierarchical structure.	
	The procedure proceeds with two main steps: (1) Partition the networks into two groups using spectral clustering on an estimated distance matrix $\hat D$ based on a graphon estimation method.
	(2) Test whether the partition is statistically significant; if so, recursively apply the procedure to each subgroup.
	%\begin{enumerate}[label=(\arabic*)]
	%\item Partition the networks into two groups using spectral clustering on an estimated distance matrix $\widehat D$ based on a graphon estimation method.
	%\item Test whether the partition is statistically significant; if so, recursively apply the procedure to each subgroup.
	%\end{enumerate}
	
	One of the advantages of this algorithm is that we do not need to determine the number of clusters beforehand, which is often unrealistic for real data applications.
	Second, compared to traditional spectral clustering, each step involves only limited eigenvectors, which greatly improves the computational efficiency.
	Finally, while we restrict our focus in this paper to binary trees, the methodology can be easily extended to other kinds of hierarchical trees.

	\begin{algorithm}
		\caption{Network Hierarchical Clustering based on Two Sample Test (NHC-TST)}
		\label{NHC-TST}
		
		\LinesNumbered
		
		\SetKwProg{Fn}{Procedure}{}{}
		\Fn{HierarchicalClustering $(\widetilde{A},\alpha):$ \tcp*[f]{$\widetilde{A} = (A^{(1)},\ldots,A^{(m)})$}}{
			
			\nl $(\widehat{P}^{(1)},\ldots,\widehat{P}^{(m)}) \leftarrow \text{NBS}(A^{(1)},\ldots,A^{(m)})$  \label{alg:est_m} 
			
			\tcp*[f]{Any link probability estimator} 
			
			$\widehat{D} \leftarrow (\widehat{P}^{(1)},\ldots,\widehat{P}^{(m)})$ 
			\tcp*[f]{$\widehat{D}_{ij} = \frac{1}{n} ||\widehat{P}^{(i)} - \widehat{P}^{(j)}||_F$} 
			
			$(\widetilde{A}_1, \widetilde{A}_2)\leftarrow\text{SpectralClustering}(\widehat{D})$ \label{alg:spc} 
			\tcp*[f]{$\widetilde{A}_i = (A^{(1)},\ldots,A^{(m_i)}),\ \  m_1 + m_2 = m$} 
			
			$({\widetilde{P}}_1,{\widetilde{P}}_2)\leftarrow \text{MNBS}(\widetilde{A}_1, \widetilde{A}_2)$ \label{alg:est_2} 
			
			$\widehat{Z}=\frac{\bar{A}_1-\bar{A}_2}{\sqrt{\left[n\left\{\frac{1}{m_{1}} {\widetilde{P}}_1\left(1-{\widetilde{P}}_1\right)+\frac{1}{m_{2}} {\widetilde{P}}_2\left(1-{\widetilde{P}}_2\right)\right\}\right]}}$
			\tcp*[f]{$\bar{A}_i = \frac{1}{m_i}\sum_{r=1}^{m_i}A^{(r)}$} 
			
			$\widehat{\theta} \leftarrow \frac{1}{\sqrt{15}} \Tr (\widehat{Z}^3)$ 
			
			\vspace{5pt} 
			\eIf{$\mathbb{P}(|\widehat{\theta}| > z_{\alpha/2}) < \alpha$ \label{alg:test}}{
				$(\widetilde{A}_1, \widetilde{A}_2)\leftarrow BIPARTITION(\widetilde{A})$ 
				
				$\text{HierarchicalClustering}(\widetilde{A}_1,\alpha)$ 
				
				$\text{HierarchicalClustering}(\widetilde{A}_2,\alpha)$ 
			}{\vspace{5pt} 
				STOP
			}
		}
		\vspace{5pt} 
		\textbf{end procedure}
	\end{algorithm}
	
	For each network $i = 1, \dots, m$,  let $P^{(i)}$ and $\hat{P}^{(i)}$ represent the underlying link probability matrix and its corresponding estimation, respectively. 
	The complete recursive procedure is formalized in Algorithm \ref{NHC-TST} for specific choices of estimation, testing, and clustering algorithms.
	In particular, any estimation algorithm could be adopted in lines \ref{alg:est_m} and \ref{alg:est_2}.
	Here, we choose the neighborhood smoothing (NBS) \citep{zhang2017estimating} and modified neighborhood smoothing  (MNBS) \citep{zhao2019change} methods, respectively.
	Similarly, for the decision rule in line \ref{alg:test}, we choose the two-sample hypothesis test for populations of networks from \citet{chen2024spectral}.
	However, other criteria could be used such as AIC or BIC, though we do not pursue their theoretical properties here.
	Finally, a crucial component of Algorithm \ref{NHC-TST} is the clustering executed in line \ref{alg:spc}.
	Given the estimated distance matrix $\hat{D} \in \mathbb{R}^{m \times m}$, we propose two distinct bipartitioning approaches based on eigenvector sign-checks.
	%Depending on the chosen matrix, we denote the resulting algorithmic variants as $\text{NHC}_{\text{L}}$ and $\text{NHC}_{\text{D}}$, formalized as follows:
	Depending on the chosen matrix, the corresponding  bipartitioning rules formalized as follows:
	
	\begin{itemize}[leftmargin=1.5em]
		\item[(1)] \textbf{Laplacian Spectral Bipartitioning ($\text{NHC}_{\text{L}}$):} This approach operates on the  Laplacian  $L_{\hat D}~=~\text{diag}(\hat D \mathbf{1}) - \hat{D}$.
		We compute the dominant eigenvector $u \in \mathbb{R}^m$ of $L_{\hat D}$ and assign the binary class label $c_i = \mathbb{I}(u_i > 0)$ for all $i = 1, \dots, m$.
		\item[(2)] \textbf{Distance Spectral Bipartitioning ($\text{NHC}_{\text{D}}$):} 	This approach operates directly on the estimated distance matrix $\hat{D}$. We compute the minimal eigenvector $u \in \mathbb{R}^m$ of $\hat{D}$ and similarly assign the class label $c_i = \mathbb{I}(u_i > 0)$.
	\end{itemize}
	
	In both cases, the sign pattern of the eigenvector determines the partition.
	For both $\text{NHC}_{\text{L}}$ and $\text{NHC}_{\text{D}}$, the generated binary indicator $c_i \in \{0,1\}$ serves to estimate the underlying hierarchical branching encoding scheme $a_i$.
	
	Combining these bipartitioning rules with the recursive two-sample testing procedure, we denote the resulting algorithmic variants as $\text{NHC}_{\text{L}}\text{-TST}$ and $\text{NHC}_{\text{D}}\text{-TST}$, respectively.
	
	\section{Theoretical results}
	\label{sec:theory}
	%We begin our theoretical analysis by considering the ideal scenario where the distance matrix $D$ is available ($D$ is viewed as a distance matrix throughout the remainder of the paper).
	%In this ideal setting, we will demonstrate that both proposed methods $\text{NHC}_{\text{L}}$ and $\text{NHC}_{\text{D}}$ can recover the exact hierarchical network structure.
	%We then show that when the true $D$ is unknown and replaced by its estimate $\hat{D}$, we have consistency of $\text{NHC}_{\text{L}}$. 
	
	Our theoretical analysis proceeds in three stages. Assuming the recursive splitting correctly terminates, we prove that both $\text{NHC}_{\text{L}}$ and $\text{NHC}_{\text{D}}$ exactly recover the underlying hierarchical bipartitions when the distance matrix $D$ is available. Then, we establish the  consistency of both methods in the empirical setting when $D$ is unknown and replaced by its estimate $\hat{D}$. Finally, by integrating the data-driven testing procedure, we establish the consistency of the $\text{NHC}_{\text{D}}\text{-TST}$ algorithm in recovering the true hierarchical topology.

	Before presenting our theoretical results, we specify the block dimensions based on the underlying hierarchical tree. For any non-terminal cluster with binary label $x$ of size $m_x$, let $x0$ and $x1$ be its child clusters with sizes $m_{x0}$ and $m_{x1}$, where $m_x = m_{x0} + m_{x1}$. Accordingly, the submatrix $D^{[x]}$ is naturally partitioned into diagonal blocks $D^{[x0]} \in \mathbb{R}^{m_{x0} \times m_{x0}}$, $D^{[x1]} \in \mathbb{R}^{m_{x1} \times m_{x1}}$, and an off-diagonal block $R^{[x]} \in \mathbb{R}^{m_{x0} \times m_{x1}}$. When $D$ is an HDM, this block partition strictly satisfies the structural conditions in Definition \ref{def:HM}.
	%The monotonic property of the distance matrix implies that this structural partition naturally exhibits a Small-Large magnitude pattern. 
	Then, it is standard and natural to adopt the following bounded assumptions on the entries of the sub-blocks:
	
	\begin{assumption}[Between-cluster distances] \label{A1}
		The cross-branch distances are bounded such that all entries in $R^{[x]}$ satisfy $R^{[x]}_{ij} \in [\alpha_1^x, \beta_1^x]$.
	\end{assumption}
	
	\begin{assumption}[Within-cluster distances] \label{A2}
		%	The within-cluster distances are strictly smaller than the between-cluster distances. Specifically,  $D^{[x0]}, D^{[x1]} \in [0, \beta_0^x]$, with gap condition $\beta_0^x < \alpha_1^x$. The smallest non-zero entry of all the diagonal blocks is $\alpha_0^x$.
		The non-zero entries of $D^{[x0]}$ and $D^{[x1]}$ are bounded within $[\alpha_0^x, \beta_0^x]$, satisfying the gap condition $\beta_0^x < \alpha_1^x$.
	\end{assumption}
	
	\begin{assumption}[Cluster sizes and balance] \label{A3}
		The sample size balance ratio between the two child clusters satisfies $\eta^x = m_{x0} / m_{x1} > 0$. Denote $\tilde{\eta}^x = \max\{\eta^x, 1 / \eta^x\}$.
	\end{assumption}
	
	In particular, for the initial split at the root of the hierarchy, there is no parent cluster index, so we omit the superscript $x$ for simplicity and denote the overall variables as $D, R, m, \alpha_0, \alpha_1, \beta_0, \beta_1$, $\eta$, and $\tilde\eta$. 
	%Moreover, from Assumptions \ref{A1} and \ref{A2}, we have $\beta_0^x = \max\{\beta_1^{x0}, \beta_1^{x1}\}$. At the root layer, this reduces to $\beta_0 = \max\{\beta_1^{0}, \beta_1^{1}\}$.
	
	%The following proposition about the class of Positive-Negative block matrix is a preliminary for our results.
	%
	%\begin{proposition}[Positive-Negative block matrix \citep{balakrishnan2011noise}] \label{Positive-Negative}
	%Let $B$ be an $m \times m$ symmetric matrix with the Positive-Negative block structure of 
	%\begin{equation}
	%	B = 
	%	\begin{pNiceArray}{c|c}[cell-space-limits = 2pt]
		%		B_+ & B_- \\
		%		\Hline
		%		B_-^T & \tilde{B}_+ \\
		%	\end{pNiceArray},
	%\end{equation}
	%where every off-diagonal element in the $p\times p$ block $B_+$ and $q \times q$ block $\tilde{B}_+$ is strictly positive and every element in the $p \times q$ block $B_-$  is strictly negative ($p \geq 1, q \geq 1$). Let $v$ be the dominant eigenvector of $B$. Then the eigenvector $v$ either has the sign pattern of $\begin{pmatrix} v_+ \\ v_- \end{pmatrix}$, where $v_+$, the first $p$ elements of $v$, are strictly positive and $v_-$, other $q$ elements of $v$, are strictly negative or has the reverse sign pattern.
	%\end{proposition}
	
	%The first theorem below establishes the exact recovery guarantee for the Laplacian-based approach $\text{NHC}_{\text{L}}$, using the Laplacian matrix $L_D$. 

	\begin{theorem}\label{theorem L_D}		
		Let $D \in \mathbb{R}^{m \times m}$ be a normalized population distance matrix, and $L_D~=~\text{diag}(D \mathbf{1}) - D$ be its associated Laplacian matrix. Suppose $D$ is an HDM. Under Assumptions \ref{A1} and \ref{A2}, the $\text{NHC}_{\text{L}}$ method will exactly recover the true bipartition at each recursive split. %true hierarchical topology.
	\end{theorem}

	\begin{corollary}
		Under the assumptions of Theorem \ref{theorem L_D}, if the splitting process correctly stops at all terminal clusters, then recursively applying $\text{NHC}_{\text{L}}$ to each newly generated sub-cluster exactly recovers the underlying hierarchical topology.
	\end{corollary}
	
	%Theorem \ref{theorem L_D} is highly general. Relying solely on Assumptions \ref{A1} and \ref{A2} and the intrinsic monotonicity of the HDM, it demonstrates that the exact recovery is driven purely by the Laplacian $L_D$, rather than any extra parametric conditions.
	Theorem \ref{theorem L_D} establishes a highly general recovery guarantee for  each recursive bipartition. Under Assumptions \ref{A1} and \ref{A2}, the exact bipartition is determined by the spectral structure of the Laplacian matrix $L_D$ without additional parametric assumptions. The corollary extends this result to the full hierarchy, highlighting the generality of the proposed framework based solely on the hierarchical structure represented  by the HDM.
	
	The following theorem below establishes a parallel guarantee for $\text{NHC}_{\text{D}}$:  Operating directly on $D$ can exactly recover  the initial root bipartition. %the true hierarchical topology.
	Without loss of generality, we consider the first split of the hierarchical clustering, since the procedure repeats down the branches of the hierarchy.
	
	\begin{theorem}\label{theorem D}
		%	For a normalized distance matrix $D \in \mathbb{R}^{m \times m}$ satisfying the monotonicity property described in Definition \ref{def:HDM}, let 
		Let $D \in \mathbb{R}^{m \times m}$ be a normalized population distance matrix and let
		\begin{equation*}
			%g(m_{0-},m_{1+}) = \frac{- \beta_0(m_{0-}^2 + m_{1+}^2) + 2 m_{1+} m_{0-} \beta_0+ \beta_0(m_{0-} - m_{1+})(m_0 - m_1)}{m_0 m_{1+} + m_1 m_{0-} - 2 m_{1+} m_{0-}},
			g(m_{0-}, m_{0+}, m_{1-}, m_{1+}) =  \frac{\beta_0 (m_{0-} - m_{1+}) (m_{0+} - m_{1-})}{ m_{0+} m_{1+} + m_{1-} m_{0-}},
		\end{equation*}
		where $m_{0+} + m_{0-} = m_0, \ m_{1+} + m_{1-} = m_1, \ 1 \leq m_{0-} \leq m_0 - 1$, $1 \leq m_{1+} \leq m_1 - 1$. Define $M~=~\max\limits_{m_{0-}, m_{0+}, m_{1-}, m_{1+}} g(m_{0-}, m_{0+}, m_{1-}, m_{1+})$. 
		Suppose $D$ is an HDM and Assumptions \ref{A1} -- \ref{A3} hold.
		Method $\text{NHC}_{\text{D}}$ will exactly recover the initial root bipartition if the following separation conditions hold:
		$$
		\text{(1)} \quad \alpha_1 > \beta_0 + M, \qquad \text{and} \qquad \text{(2)} \quad \alpha_1 > \max\{\eta, \eta^{-1}\} \beta_0.
		$$
		%\begin{enumerate}[label=(\arabic*)]
		%	\item $\alpha_1 > \beta_0 + M \enskip.$ 
		%	\item $\alpha_1 > \max\left\{\eta, 1 / \eta \right\} \beta_0 \enskip.$
		%\end{enumerate}
		
		%Furthermore, provided that analogous local separation conditions are satisfied at all internal nodes, the recursive application of $\text{NHC}_{\text{D}}$ will exactly recover the true hierarchical topology.
	\end{theorem}
	
	The following corollary extends this result recursively to recover the full hierarchical topology.
	\begin{corollary}
		Under the assumptions of Theorem \ref{theorem D}, suppose that the analogous separation conditions are satisfied at every recursive split. 
		Then, by recursively applying Theorem \ref{theorem D} to each newly generated sub-cluster, $\text{NHC}_{\text{D}}$ correctly recovers every internal split. 
		Furthermore, if the splitting process correctly stops at all terminal clusters, then $\text{NHC}_{\text{D}}$ exactly recovers the underlying hierarchical topology.
	\end{corollary}

	The two conditions in Theorem \ref{theorem D} provide both an absolute and a relative threshold on the signal strength required for exact recovery of a bipartition.  
	The first condition means the absolute gap $\delta~=~\alpha_1 - \beta_0$ must overcome a penalty $M \ge 0$ arising from cluster size imbalance. For balanced clusters, $M = 0$, any positive gap suffices; for unbalanced clusters, $M > 0$, a larger separation between within- and between-cluster connectivity is needed. This theoretically explains why classic spectral methods naturally favor balanced partitions and struggle with unbalanced ones.
	The second relative threshold $\alpha_1 > \max\left\{\eta, 1 / \eta \right\} \beta_0$ links the required signal strength to cluster size imbalance. For balanced clusters where $\eta = 1$, it reduces to $\alpha_1 > \beta_0$, implied by Assumptions \ref{A1} and \ref{A2}. For unbalanced clusters, within- and between-cluster connectivity needs to be distinguishable enough to compensate for the structural disparity in subset sizes.
	
	Since the true matrix $D$ is generally unavailable in practice, theoretical guarantees must be extended to the estimated matrix $\hat D$. In this paper, we compute it from the link probability matrix estimated via the NBS method. Before establishing the $\text{NHC}_{\text{L}}$ framework when $D$ is replaced by $\hat D$ in a single split setting, we introduce the necessary regularity assumptions.

			\begin{assumption}[Estimation error] \label{asum 7}
				Let $\{\hat P^{(i)}\}_{i = 1}^m$ be the estimators of the true link probability matrices $\{P^{(i)}\}_{i = 1}^m$. We assume that for any $\varepsilon > 0$, there exists an error rate $\zeta$ such that
				\begin{equation*}
					P \left( \max_{1 \leq i \leq m} \frac{1}{n} \|\hat P^{(i)} - P^{(i)}\|_F \leq \zeta \right) \geq 1 - n^{-\varepsilon} \enskip,
				\end{equation*}
				where $\zeta$ depends on both $n$ and $\varepsilon$.
			\end{assumption}
			
			\begin{assumption}[Separation margin] \label{asum 4} 
				Define $\Delta := \alpha_1 - \frac{1}{1 + \tilde \eta} (\beta_0 + \tilde \eta \beta_1)$. Assume that $\Delta > 0$ and  $\Delta = \omega(\zeta)$ as $n \to \infty$. 
			\end{assumption}
			
			%\begin{assumption}[Separation] \label{asum 4} 
			%	$\Delta := \alpha_1 - (\beta_0 + \tilde \eta \beta_1) / (1 + \tilde \eta) = \omega \big( \sqrt{\mathstrut n \log n} / n \big)$, where $\Delta > 0$. 
			%\end{assumption}
			
			\begin{assumption}[Bounded between-cluster variability]
				$\beta_1 = O(\Delta)$. \label{asum 5}
			\end{assumption}
			
			\begin{assumption}[Eigenvalue-degree condition]
				$\lambda_1(L_D) > \deg_i(D)$ for all $i$.  \label{asum 6}
			\end{assumption}
			
			%\begin{assumption}[Graphon estimation error] \label{asum 7}
			%	Let $\{\hat P^{(i)}\}_{i = 1}^m$ be the estimated link probability matrices of $\{P^{(i)}\}_{i = 1}^m$ applying NBS method.  $D$ and $\hat D$ are the true population distance matrix  and the estimation built from $\{\hat P^{(i)}\}_{i = 1}^m$. %Consequently, \eqref{NBS} and Lemma \ref{lemma DK} hold. \label{asum 7}
			%\end{assumption}
			
			The following theorem guarantees the consistency of $\text{NHC}_{\text{L}}$ framework under these settings.
			
			\begin{theorem}\label{theorem sign}
				For a normalized population distance matrix $D \in \mathbb{R}^{m \times m}$, suppose that $D$ is an HDM and satisfies Assumptions \ref{A1}--\ref{asum 6}.
				%\begin{enumerate}[label = \textbf{(A\arabic*)}, start = 4]
				%	\item \textbf{Separation:} $\alpha_1 - \frac{1}{1 + \tilde \eta} (\beta_0 + \tilde \eta \beta_1) > 0$ defined in Lemma \ref{lemma HDM} satisfies $\Delta = \omega(\frac{ \sqrt{\mathstrut n \log n}}{n} ) \enskip.$  \label{asum 4} 
				%	\item \textbf{Bounded between-cluster variability:} $\beta_1 = O(\Delta) \enskip.$ \label{asum 5}
				%	\item \textbf{Eigenvalue-degree condition:} $\lambda_1(L_D) > \deg_i(D)$ for all $i$ \enskip.  \label{asum 6}
				%	\item \textbf{Graphon estimation error:} Let $\{\hat P^{(i)}\}_{i = 1}^m$ be the estimated link probability matrices of $\{P^{(i)}\}_{i = 1}^m$ applying NBS method from \citet{zhang2017estimating}.  $D$ and $\hat D$ are the true population distance matrix  and the estimation built from $\{\hat P^{(i)}\}_{i = 1}^m$. Consequently, \eqref{NBS} and Lemma \ref{lemma DK} hold. \label{asum 7}
				%\end{enumerate}
				Let $u$ and $\hat u$ be the dominant unit eigenvector of $L_D$ and $L_{\hat D}$, respectively. Then, for any $\varepsilon > 0$  and sufficiently large enough $n$, we have % and $m$,
				\begin{equation}\label{sign}
					P(\text{sign}(\hat u) = \text{sign}(u)) \geq (1 - n^{- \varepsilon})^m \enskip. 
				\end{equation}		
				In particular, if there exists $\varepsilon > 0$ such that $m / n^\varepsilon = o(1)$ as $n \to \infty$, where $m$ may either remain fixed or grow with $n$, Method $\text{NHC}_{\text{L}}$ asymptotically recovers the sign pattern, i.e., $\text{sign}(\hat u)~=~\text{sign}(u)$ with probability tending to 1.
			\end{theorem}
			
			%\begin{remark}
			%Theorem \ref{theorem sign} shows that the spectral mechanism remains robust to the estimation error in $\hat{D}$. Since the NBS estimator incurs a normalized error rate of order $\sqrt{\log n / n}$,  Assumption \ref{asum 4} requires the graph-level cluster separation to dominate this noise level, ensuring that between-cluster distances remain distinguishable from within-cluster distances.   In particular, if the graph separation is bounded away from zero, so that $\Delta$ converges to a positive constant, then Assumption \ref{asum 4} holds automatically.
			%
			%\hfill $\square$
			%\end{remark}
			
			\begin{remark}
				Theorem \ref{theorem sign} guarantees the reliability of each recursive split by showing  that the empirical sign pattern $\text{sign}(\hat{u})$ obtained by $\text{NHC}_{\text{L}}$ consistently recovers the population pattern $\mathrm{sign}(u)$, thereby identifying the underlying bipartition asymptotically.	
				%Theorem \ref{theorem sign}
				This theorem does not rely on a specific construction of the link probability matrix estimators. Instead, it applies to any estimators satisfying the uniform error bound in Assumption \ref{asum 7}.
				Assumption \ref{asum 4} requires that the population-level cluster separation $\Delta$ dominates the estimation error level $\zeta$. Together, these conditions guarantee that the underlying hierarchical structure remains identifiable despite estimation uncertainty.
				In this paper, we employ the NBS-type estimator \citep{zhao2019change} for each link probability matrix. The NBS estimator satisfies
				\begin{equation*}
					\max_{1 \leq i \leq m} \frac{1}{ n}\|\hat P^{(i)} - P^{(i)} \|_F  \leq C_0  \sqrt{\frac{\log n} {n} }\enskip,
				\end{equation*}
				with probability at least $1-n^{-\varepsilon}$ for any $\varepsilon>0$, where $C_0$ is a positive global constant depending on $\varepsilon$ and another global constant $B_0 > 0$ but not on $n$ or $m$. Hence Assumption \ref{asum 7} holds with $\zeta = C_0 \sqrt{\frac{\log n} {n} }$. 				
				Under this specific rate, Assumption \ref{asum 4} reduces to $\Delta = \omega\big(\sqrt{\frac{\log n} {n} } \big)$. In particular, if the graph separation is bounded away from zero, so that $\Delta$ converges to a positive constant, this condition is automatically satisfied, and the recovery guarantee of Theorem \ref{theorem sign} follows.
			\end{remark}

			Building on the previous theoretical results, we next establish the consistency of $\text{NHC}_{\text{L}}\text{-TST}$ under the proposed two-sample testing stopping rule.
			
			\begin{theorem} \label{theorem consistency}
				Let $D \in \mathbb{R}^{m \times m}$ be a normalized population distance matrix.  Suppose that $D$ is an HDM and satisfies assumptions \ref{A1}--\ref{asum 6}. 
				%	Let $D$ be an HDM among $m$ networks satisfying assumptions \ref{A1}-\ref{asum 7}. 
				Assume the true underlying hierarchical tree $\Gamma$ has a finite maximum depth of $\mathcal L \geq 0$. We index the clusters at any layer $l \in \{0, 1, \dots, \mathcal L\}$ using binary strings $x \in \{0, 1\}^l$, adopting the convention that the root node corresponds to the empty string.  Let $\delta_x$ be the indicator variable for cluster $x$, where 
				\begin{equation*}
					\delta_x = 
					\begin{cases}
						0 & \text{if cluster $x$ is terminal with no further splitting}, \\
						1 & \text{if cluster $x$ is an internal node requiring bipartition}.
					\end{cases}		
				\end{equation*}
				%Additionally, assume that the cluster size $m_x$  satisfies $m_x = o(n^\varepsilon)$ as $n \to \infty$ for some  $\varepsilon > 0$, allowing $m_x$ to be either fixed or increase with $n$. Furthermore, suppose the nominal significance level $\alpha$ for the stopping rule dynamically satisfies $\alpha = o(2^{- \mathcal L})$.
				Additionally, let $\beta_x$ denote the Type II error rate of the splitting procedure for cluster $x$, and let $\alpha_n$ denote the nominal significance level of the stopping rule.
				
				Then, the $\text{NHC}_{\text{L}}\text{-TST}$ method recovers the true underlying structure $\Gamma$ exactly with probability at least 
				\begin{equation*}
					1 - \sum_{l = 0}^{\mathcal L} \sum_{\substack{x \in \{0,1\}^l \\ \delta_x = 1}} (m_x n^{- \varepsilon} +  \beta_x)
					- 2^{\mathcal L}  \alpha_n \enskip.
				\end{equation*}
				%where $\beta_{x}$ denotes the Type II error rate  of the splitting procedure to cluster $x$.
				
				%In particular, the probability of correct recovery will tend to 1 as $n \to \infty$ and $m_x = o(n^\varepsilon)$.
				In particular, with the nominal significance level chosen such that $\alpha_n = o(1)$, if the embedded two-sample test achieves asymptotic power one and the cluster size $m_x$ satisfies $m_x = o(n^\varepsilon)$ for some $\varepsilon > 0$, then correct recovery is achieved with probability tending to $1$ as $n \to \infty$.
			\end{theorem}

			Among the requirements for the recovery guarantees in Theorem \ref{theorem consistency}, two key conditions are particularly notable. First, the sub-group sizes satisfy $m_x = o(n^\varepsilon)$, which flexibly allows $m_x$ to be either fixed or diverging with $n$. Second, the embedded two-sample test is asymptotically powerful. 				
			Crucially, both conditions are automatically satisfied by incorporating the MNBS-based two-sample testing procedure of \citet{chen2024spectral} (Algorithm \ref{NHC-TST}, lines \ref{alg:est_2}--\ref{alg:test}). 
			On the one hand, their theoretical framework requires the sub-network size to be bounded by $O(n^{\alpha_u})$ for some constant $\alpha_u > 0$ to ensure asymptotic normality under the null hypothesis. Given our condition $m_x = o(n^\varepsilon)$ for some $\varepsilon > 0$, this requirement automatically holds by taking $ \alpha_u \geq \varepsilon$. 
			On the other hand, Theorem 3 of their work formally establishes the consistency of the corresponding test under mild regularity conditions, thereby guaranteeing the required asymptotic power. Consequently, adopting their procedure naturally validates the exact theoretical requirements for our asymptotic results.
			
			\section{Simulation study}
			\label{sec:sim}
			
			In this section, we evaluate the empirical performance of the two proposed variants of our \textbf{NHC-TST} method: the distance-based (\textbf{$\text{NHC}_{\text{D}}\text{-TST}$}) and Laplacian-based (\textbf{$\text{NHC}_{\text{L}}\text{-TST}$}) approaches. Throughout the simulations, the nominal significance level is set to $\alpha_n = 1/n$. We compare the proposed methods with four state-of-the-art competitors: \textbf{HICL-SBM} \citep{rebafka2024model}, the alternating minimization algorithm (\textbf{ALMA}) for mixed multilayer SBMs in \cite{fan2022alma}, \textbf{NCGE/NCLM} \citep{mukherjee2017clustering}, and \textbf{NSBM} \citep{josephs2023nested}.
			%The proposed methods are compared with several state-of-the-art competitors, whose details and parameter settings are described below.
			
			HICL-SBM is implemented using the \texttt{graphclust} R package with the maximum number of SBM blocks set to $Q_{\max}=\lceil\log n\rceil$. For ALMA and NCGE/NCLM, the true number of network clusters $K$ is provided. In SBM settings, we set a uniform community structure across all clusters and supply ALMA with the true number of node communities $Q$.
			Following \citet{fan2022alma}, we select $Q$ from $\{2, 3, 4\}$ in non-SBM settings by applying the elbow method based on the mean squared error (MSE) criterion.
			%In non-SBM settings, $Q$ is selected from $\{2,3,4\}$ using the elbow method based on the mean squared error criterion, following \citet{fan2022alma}. 
			For NCLM, the maximum moment order is fixed at $10$. For NSBM, to facilitate posterior sampling via finite truncation approximations, we set the maximum truncation levels for the network classes and communities of nodes within class to $K_{\mathrm{trunc}}=Q_{\mathrm{trunc}}=10$, and the Gibbs sampler is run for $200$ iterations.

			In our simulations, we assume a balanced design where the cluster sizes are set to be equal, with $m_k = m^*$ for all $k = 1, \ldots, K$, where $m^* \in \{5, 10, 15, 20\}$.
			The number of nodes for each network varies over $n \in \{100, 150, 200, 300\}$. We evaluate the competing methods comprehensively across three key dimensions: 
			\begin{enumerate}[label=(\arabic*)]
				\item \textbf{Partitioning accuracy}: We evaluate the clustering results against the ground-truth labels using normalized mutual information (NMI) \citep{strehl2002cluster}. Additionally, we report the estimated number of clusters for each method to assess model-selection behavior. Note that ALMA and NCGE/NCLM are excluded from this comparison of cluster-number estimation, since they take the true $K$ as input, while other methods learn the number of network clusters automatically.
				\item \textbf{Computational efficiency}:  We track and report the $\log_{10}$-transformed total execution time (in seconds) for each method to complete the clustering task.
				\item \textbf{Hierarchical topology recovery}: To further assess the recovery of the underlying hierarchical structure, we additionally measure performance using a novel adaptation of the Cophenetic Correlation Coefficient (CPCC), which we detail below.
			\end{enumerate}
			
			While the original CPCC \citep{sokal1962comparison} evaluates an estimated dendrogram by correlating its cophenetic distances with the initial pairwise data distances, it is a surrogate necessary when the true underlying tree structure is unavailable. However, since our simulation setting provides the ground-truth hierarchy, we propose an adapted CPCC that directly correlates the estimated structural distances with the true ones.
			
			Specifically, the true topological distance between any pair of sample networks is measured using the tree distance defined in \eqref{eq:tree-dist}, reflecting the hierarchical depth at which the splits occur.	
			To ensure a fair comparison, the estimated pairwise distances are computed differently based on each algorithm's output. For our NHC-TST variants, the estimated distances are computed using the  same tree distance formulation \eqref{eq:tree-dist} on the estimated hierarchy. In contrast, for HICL-SBM, the estimated distances are represented by standard cophenetic distances derived from its output dendrogram. The adapted CPCC is then calculated as the Pearson correlation coefficient between the flattened true and estimated distance matrices. 
			Because ALMA, NCGE, NCLM, and NSBM output flat clusterings, we compute the adapted CPCC only for our proposed methods ($\text{NHC}_{\text{L}}\text{-TST}$, $\text{NHC}_{\text{D}}\text{-TST}$) and HICL-SBM.
			Across all evaluated methods, an adapted CPCC score approaching $1$ indicates a highly accurate topological reconstruction.
			
			\begin{remark}
				Notably, while our method, like NCGE, utilizes distance between the underlying graphons, it differs in how the number of clusters is determined. Instead of assuming a fixed number of clusters as in standard spectral clustering and NCGE, we integrate a two-sample test into the hierarchical procedure. This provides a data-driven stopping criterion that automatically determines the appropriate number of clusters.
				
				In addition, our NHC-TST variants generate a top-down tree topology where splits at each hierarchical level are based on a two-sample test and recovery of the underlying hierarchy is provably consistent.
				In contrast, HICL-SBM employs a hierarchical agglomerative algorithm that naturally produces a dendrogram.
				However, this ``inferred" tree primarily serves as an algorithmic trace of its greedy optimization toward a flat clustering.
				That is, its bottom-up agglomeration intrinsically forces binary merges, which can result in a highly fragmented hierarchy where intermediate levels lack explicit statistical validation, and hence no meaningful relation the latent hierarchy.
				%
				%\hfill $\square$
			\end{remark}

			\subsection{Hierarchical smooth graphons}
			In the first example, we consider a general structure of graphons. We adopt the standard definition of a graphon $f$  as follows:
			
			\begin{definition}[Graphon \citep{zhang2017estimating}]
				For any network with a link probability $P$ and number of nodes $n$, there exists a function $f: [0, 1] \times [0,1] \to [0, 1]$ and a set of i.i.d. random variables $\xi_i \sim \text{Uniform}[0, 1]$, such that 
				$P_{ij} = f(\xi_i, \xi_j), \quad i,j = 1, \ldots, n.$
				%with $i,j = 1, \ldots, n$.
			\end{definition}
			
			To simulate a network with a hierarchical structure, we set the depth of the hierarchy to $\cL = 2$, with four distinct clusters $00,\; 01,\; 10,\; 11$, i.e., $K = 4$. Consequently, each cluster is identified by a binary label  $x=a_1 a_2$, where $a_{\l} \in \{0, 1 \}$ for $\ell = 1,2$.
			Given a cluster label $x$, we construct a hierarchical smooth graphon utilizing sign variables $\mathcal{s}_\ell(x)\in\{\pm1\}$ defined as
			\begin{equation}\label{s(x)}
				\mathcal{s}_\ell(x) =
				\begin{cases}
					+1, & a_\ell=1,\\
					-1, & a_\ell=0.
				\end{cases}
			\end{equation}
			We model cluster with binary label $x$ by the following smooth graphon:
			\begin{equation*}
				G_x(u,v) = \mu + w_1 \mathcal{s}_1(x)\phi_1(u,v) + w_2 \mathcal{s}_2(x)\phi_2(u,v),
			\end{equation*}
			where $\mu = 0.5, w_1 = 0.25, w_2 = 0.12, u  \sim \text{Uniform}[0, 1]$, and $v \sim \text{Uniform}[0, 1]$. The basic functions $\phi_1(u,v)$ and  $\phi_2(u,v)$ are defined as:
			\begin{align*}
				\phi_1(u,v)  =  u v, \quad
				\phi_2(u,v)  = \sin(2  \pi  u)  \sin(2  \pi  v).
			\end{align*}
			
			\begin{figure}[!htb]
				\centering
				\includegraphics[width=\linewidth]{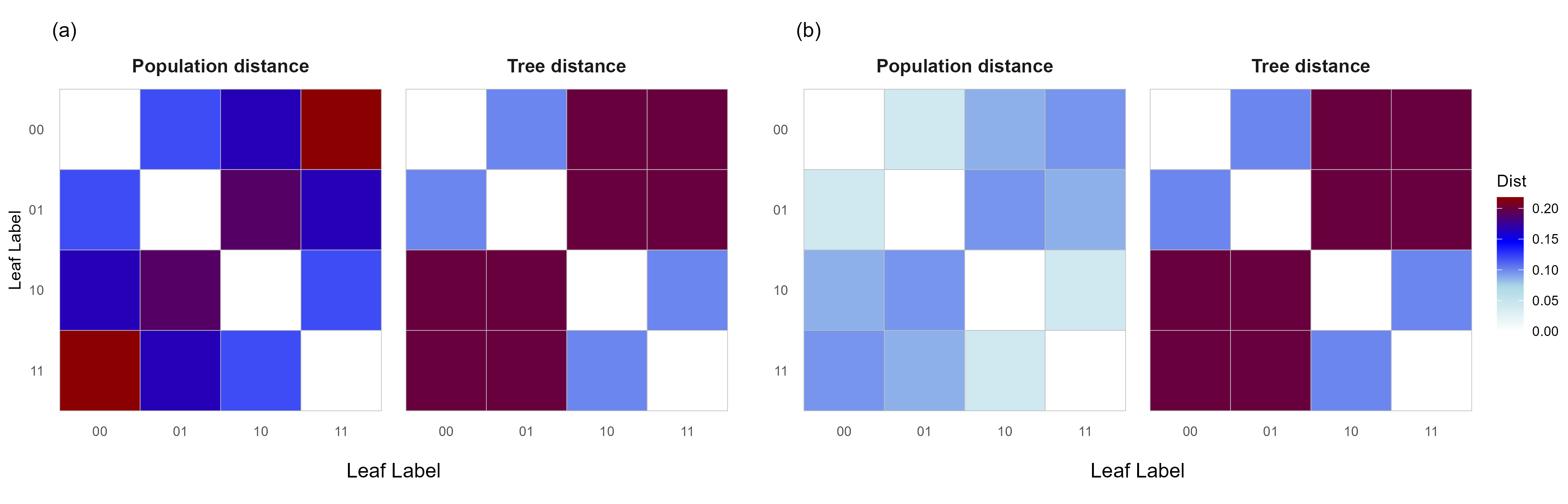}	
				\caption{Hierarchical monotonicity heatmaps (dense regime, $n = 150$). (a) Hierarchical smooth graphons. (b) Hierarchical SBMs. Within each subfigure, the left panel presents the block-level population distance matrix $\mathcal{D}$, and the right panel shows the tree distance matrix for clusters. }
%				\small \textbf{Alt text:} Heatmaps (dense regime, n = 150) comparing (a) Hierarchical smooth graphons and (b) Hierarchical SBMs. For each model, the left panel displays block-level population distance matrix, while the right panel exhibits tree distance matrix for the four clusters. In both models, the matrices exhibit consistent block-diagonal monotonicity.
				\label{fig:monotonicity}
			\end{figure}
			
			We first show that the population distance matrix $D = Z \mathcal{D} Z^\top$ in this simulation is an HDM according to Definition \ref{def:HDM}. Specifically, for any two clusters $k$ and $k'$ with corresponding binary labels $x$ and $x'$, the block-level distance is given by $\mathcal{D}_{k, k'} = \frac{1}{n} \|G_x - G_{x'}\|_F$. This core matrix $\mathcal{D}$ uniquely determines the topology of $D$. Figure~\ref{fig:monotonicity}(a) visualizes $\mathcal{D}$ (left panel) and compares it with the tree distance defined in \eqref{eq:tree-dist} (right panel). As shown, whether viewed globally or within child sub-matrices (e.g., $\{00, 01\}$ or $\{10, 11\}$), the diagonal block distances are  strictly smaller than the off-diagonal ones. This structural alignment confirms that $\mathcal{D}$ captures the ideal tree topology. Consequently, the fully expanded node-level matrix $D$ inherits this HDM property.
			%\begin{figure}[!htb]
			%	\centering
			%	\includegraphics[width=\linewidth]{Figs/monotonicity-graphon.png}	
			%	\caption{Hierarchical monotonicity heatmaps (dense regime, $n = 150, m^* = 10$). Left: Block-level population distance matrix $\mathcal{D}$ from graphon parameters. Right: Tree distance matrix for clusters. }
			%	\small \textbf{Alt text:} Side-by-side heatmaps for a dense regime where n = 150 and m-star = 10. The heatmaps use color gradients to indicate the hierarchical structure.  The left panel displays block-level population distance matrix  from graphon parameters, while the right panel exhibits tree distance matrix for the four clusters. The block-diagonal patterns in both matrices share the same monotonicity.
			%	\label{fig:monotonicity-graphon}
			%\end{figure}
			
			The clustering results for all methods are displayed in the 1st and 3rd rows in  Figure~\ref{fig:graphon-res}.
			To further investigate the performance of the clustering algorithms for sparser networks, we employ the same setting, but with the coefficients   $\mu$ and $\{w_1, w_2\}$ scaled by $\rho_1 = 6 \log(n) / n$  and $\rho_2 = 18 \log(n) / n$, respectively. The corresponding results are illustrated in the 2nd and 4th rows in Figure~\ref{fig:graphon-res}. 
			
			\begin{figure}[!htb]
				\centering
				\includegraphics[width=\linewidth]{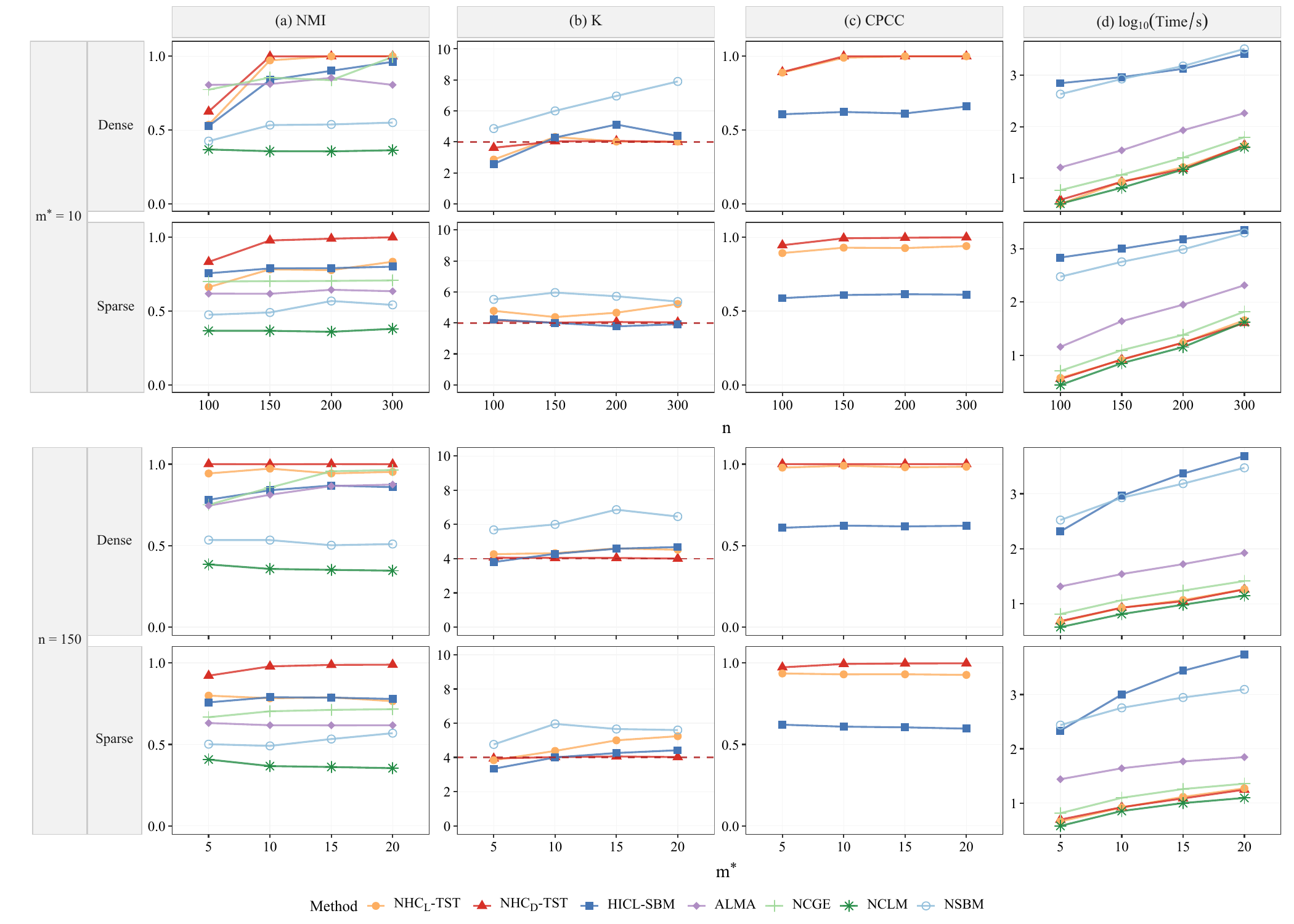}	
				\caption{Behavior comparison when the underlying networks are generated from hierarchical smooth graphons. Panels report (a)  NMI; (b) the estimated total number of clusters; (c) CPCC for $\text{NHC}_{\text{L}}\text{-TST}$, $\text{NHC}_{\text{D}}\text{-TST}$, and HICL-SBM; and (d) log transformed running time (in seconds).}
%				\small \textbf{Alt text:} A grid of line plots evaluating various clustering methods on networks generated from hierarchical smooth graphons. The plots depict NMI, estimated cluster number K, CPCC, and log 10 running time under dense and sparse conditions  as the number of networks or nodes number varies.
				\label{fig:graphon-res}
			\end{figure}
			
			Regarding clustering accuracy NMI and the estimated number of clusters $K$, the proposed NHC-TST framework exhibits robust and superior performance. In dense settings, both $\text{NHC}_{\text{D}}\text{-TST}$ and $\text{NHC}_{\text{L}}\text{-TST}$  clearly dominate, achieving high NMI and closely recovering the true $K=4$. Under sparsity, $\text{NHC}_{\text{D}}\text{-TST}$ maintains this significant advantage, while $\text{NHC}_{\text{L}}\text{-TST}$ and HICL-SBM perform comparably, remaining highly competitive.
			In extreme low-information scenarios ($n~=~100$ with $m^*=10$), ALMA, NCGE, and HICL-SBM hold a marginal lead. However, NHC-TST rapidly surpasses them as the number of nodes $n$ increases. 
			Furthermore, both NHC-TST variants capture the underlying hierarchical topology. They consistently yield CPCC scores approaching $1$ across all settings, maintaining a substantial advantage over HICL-SBM.
			Computationally, the NHC-TST framework is highly efficient, ranking second only to NCLM. Notably, while HICL-SBM yields respectable accuracy in some sparse conditions, it incurs the highest computational cost. In contrast, NHC-TST achieves superior structural recovery and competitive accuracy with striking computational efficiency.
			%comparable or better clustering accuracy while running substantially faster.
			
			\subsection{Hierarchical SBMs}
			In the second simulation, we consider a hierarchical SBM generated by adding signed perturbations to a baseline 3-block SBM. Specifically, the baseline connectivity matrix $\mathcal{B}_0$ has equal block proportions.
			We further define two perturbation matrices, $\mathcal{B}_1$ and $\mathcal{B}_2$, defined as follows:
			\[
			\mathcal{B}_0=
			\begin{pmatrix}
				0.35 & 0.15 & 0.08 \\
				0.15 & 0.15 & 0.01 \\
				0.08 & 0.01 & 0.60
			\end{pmatrix}, \quad
			\mathcal{B}_1=
			\begin{pmatrix}
				b_1 & 0 & 0 \\
				0 & 0 & 0 \\
				0 & 0 & 0
			\end{pmatrix}, \quad
			\mathcal{B}_2=
			\begin{pmatrix}
				0 & 0 & 0 \\
				0 & b_2 & 0 \\
				0 & 0 & 0
			\end{pmatrix}.
			\]
			
			We set the perturbation parameters to $(b_1,b_2) = (0.13,0.062)$. Given a cluster with binary label $x$, the cluster-specific SBM connectivity matrix is defined as: 
			\begin{equation*}
				\mathcal{B}_x = \mathcal{B}_0 + \mathcal{s}_1(x)\mathcal{B}_1 + \mathcal{s}_2(x)\mathcal{B}_2 \enskip,
			\end{equation*}
			where the sign functions $\mathcal{s}_\ell(x)$ for $\ell\in\{1,2\}$ are defined in~\eqref{s(x)}.
			We set $\cL = 2$ and generate networks from four distinct clusters:  $00,\; 01,\; 10,\; 11$, i.e., $K = 4$.
			
			The population distance matrix again forms an HDM. Similarly with the first simulation, replacing graphon blocks with SBM blocks yields $\mathcal{D}_{k, k'} = \frac{1}{n} \|\mathcal{B}_x - \mathcal{B}_{x'}\|_F$. Figure~\ref{fig:monotonicity}(b) shows that the resulting core matrix $\mathcal{D}$ remains strictly monotone and consistent with the ideal tree topology, confirming that the HDM property also holds in this SBM case.
			%\begin{figure}[!htb]
			%	\centering
			%	\includegraphics[width=\linewidth]{Figs/monotonicity-SBM.png}	
			%	\caption{Hierarchical monotonicity heatmaps (dense regime, $n = 150, m^* = 10$). Left: Block-level population distance matrix $\mathcal{D}$ from SBM parameters. Right: Tree distance matrix for clusters.}
			%	\small \textbf{Alt text:} Side-by-side heatmaps for a dense regime where n = 150 and m-star = 10. The heatmaps use color gradients to indicate the hierarchical structure. The left panel displays block-level population distance matrix from SBM parameters, while the right panel exhibits tree distance matrix for the four clusters. The block-diagonal patterns in both matrices share the same monotonicity.
			%	\label{fig:monotonicity-SBM}
			%\end{figure}		
			
			As before, we employ the same setting, but with $\mathcal{B}_x$ scaled by $\rho_2 = 18 \log(n) / n$ to assess the performance for sparser networks.
			Figure~\ref{fig:SBM-res} summarizes these results. The outcomes for the dense setting are displayed in the 1st and 3rd rows, while the results for sparser networks are presented in the 2nd and 4th rows.
			\begin{figure}[htbp]
				\centering
				\includegraphics[width=\linewidth]{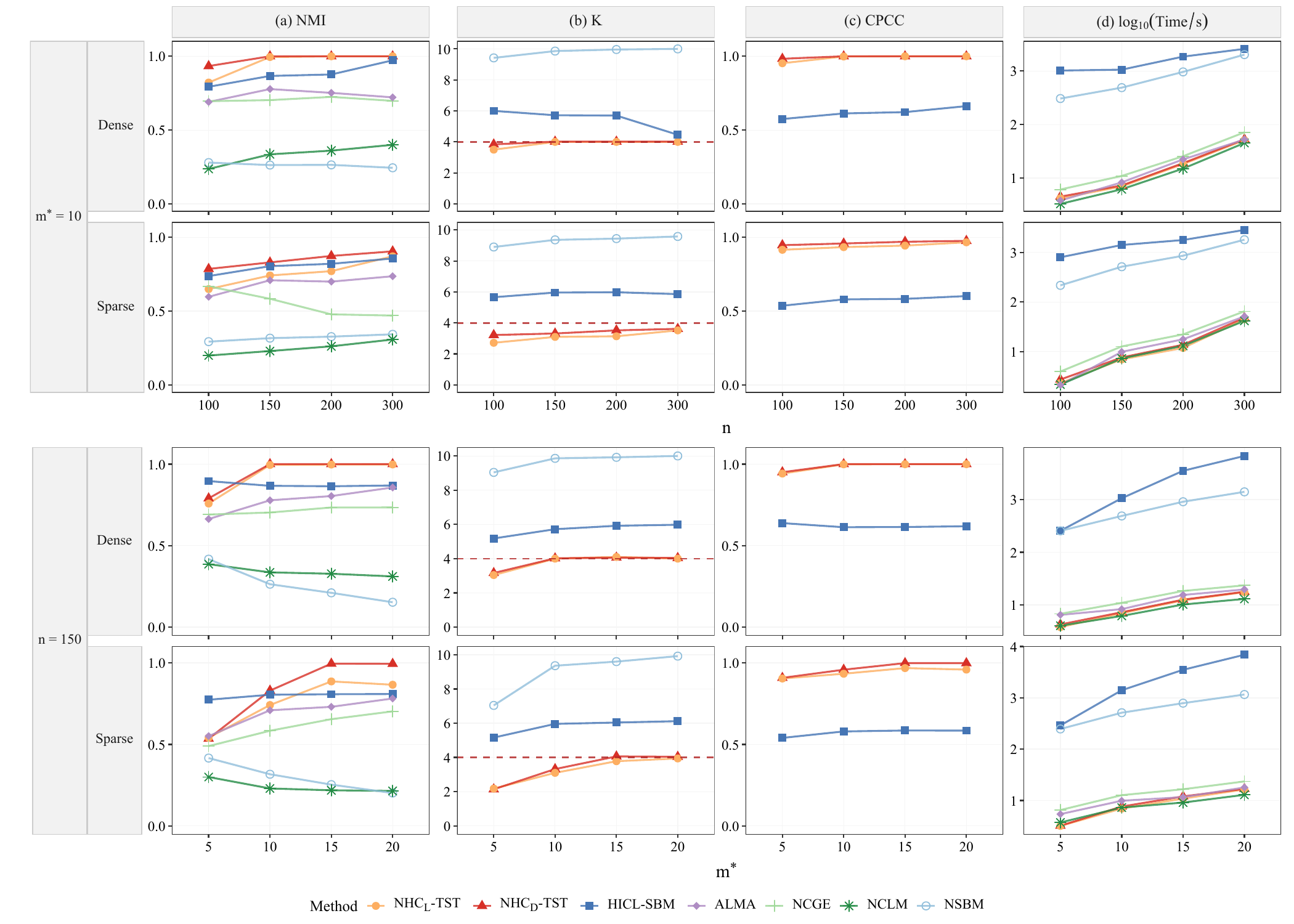}	
				\caption{Behavior comparison when the underlying networks are generated from SBMs. Panels report (a)  NMI; (b) the estimated total number of clusters; (c) CPCC for $\text{NHC}_{\text{L}}\text{-TST}$, $\text{NHC}_{\text{D}}\text{-TST}$, and HICL-SBM; and (d) log transformed running time (in seconds).}
%				\small \textbf{Alt text:} A grid of line plots evaluating various clustering methods on networks generated from SBMs. The plots depict NMI, estimated cluster number K, CPCC, and log 10 running time under dense and sparse conditions  as the number of networks or nodes number varies.
				\label{fig:SBM-res}
			\end{figure}	
			
			The results largely align with the first simulation. Regarding NMI and $K$, the proposed NHC-TST variants perform the best across most scenarios. 
			Although HICL-SBM slightly outperforms $\text{NHC}_{\text{L}}\text{-TST}$ in NMI when $m^* = 5$, as well as in sparse cases with $m^*=10$, it overestimates the true number of clusters.
			Moreover, the performance of NHC-TST variants improves as $n$ or $m^*$ increases, with the estimated cluster number tightly matching the true value of $4$. 
			Furthermore, the CPCC scores for both NHC-TST variants consistently approach $1$.
			HICL-SBM is again the most time-consuming, while our methods preserve high efficiency, matching the speed of NCLM, NCGE, and ALMA.

			\subsection{Non-uniform hierarchical smooth graphons}
			In the last simulation, we generate multiple networks from a hierarchical family of smooth graphons with non-uniform tree structure. The hierarchy has depth $\mathcal{L}=3$, allowing uneven splits across layers. For each cluster with binary label $x$, we define a cluster-specific graphon:
			\begin{equation*} 
				G_x(u,v) = \mu + w_1 \mathcal{s}_1(x)\phi_1(u,v) + w_2 \mathcal{s}_2(x)\phi_2(u,v) + w_3 \mathcal{s}_3(x)\phi_3(u,v),  
			\end{equation*}
			with $\mu = 0.5$, $\{w_1, w_2, w_3\} = \{0.32, 0.16, 0.04\}$, and $\mathcal{s}_{\cl}(x)$ for $\cl = 1, 2, 3$ as defined in \eqref{s(x)}. We truncate the values $G_x$ to the interval $[0.01,0.99]$ to ensure valid edge probabilities. 
			Here $u,v\overset{\mathrm{i.i.d.}}{\sim}\mathrm{Unif}[0,1]$, and the smooth basic functions are
			\begin{align*}
				\phi_1(u,v) = uv, \quad
				\phi_2(u,v) = \cos(\pi  (u - v)),  \quad
				\phi_3(u,v) = \sin(2 \pi u)\,\sin(2 \pi v).
			\end{align*}
			%The coefficients $(0.32,0.16,0.04)$ enforce a decreasing signal strength across levels, reflecting that deeper hierarchical distinctions are progressively subtler.
			
			We restrict our analysis to six specific clusters ($K = 6$) with the following binary labels: $000,\; 001,\; 01,\; 100,\; 101,\; 11$.
			%This design produces a non-uniform hierarchical partition, allowing us to evaluate whether the clustering methods can recover both fine- and coarse-grained branches of the underlying tree.
			%Similar to the previous simulations, the population distance matrix here strictly maintains the HDM structure, but the corresponding heatmap is omitted for conciseness.
			This design produces a non-uniform hierarchical partition. Moreover, the population distance matrix exhibits a mild departure from the  HDM monotonicity. Specifically, distances between clusters with binary labels starting with 0 and 1 are not uniformly larger than the corresponding within-group distances, while the sub-clusters within each group still satisfy the HDM monotonicity. These settings allow us to evaluate the ability of the clustering methods to recover non-uniform cases and their robustness to mild violations of the HDM assumption.
			
			The corresponding clustering results are illustrated in the 1st and 3rd rows in Figure~\ref{fig:trimmed-res}.
			Similarly, we scale the coefficients $\mu$ and $\{w_1, w_2, w_3\}$ by $\rho_1 = 6 \log n / n$ and $\rho_2 = 18 \log n / n$. 
			The corresponding results are illustrated in the 2nd and 4th rows in Figure~\ref{fig:trimmed-res}.
			
			\begin{figure}[htbp]
				\centering
				\includegraphics[width=\linewidth]{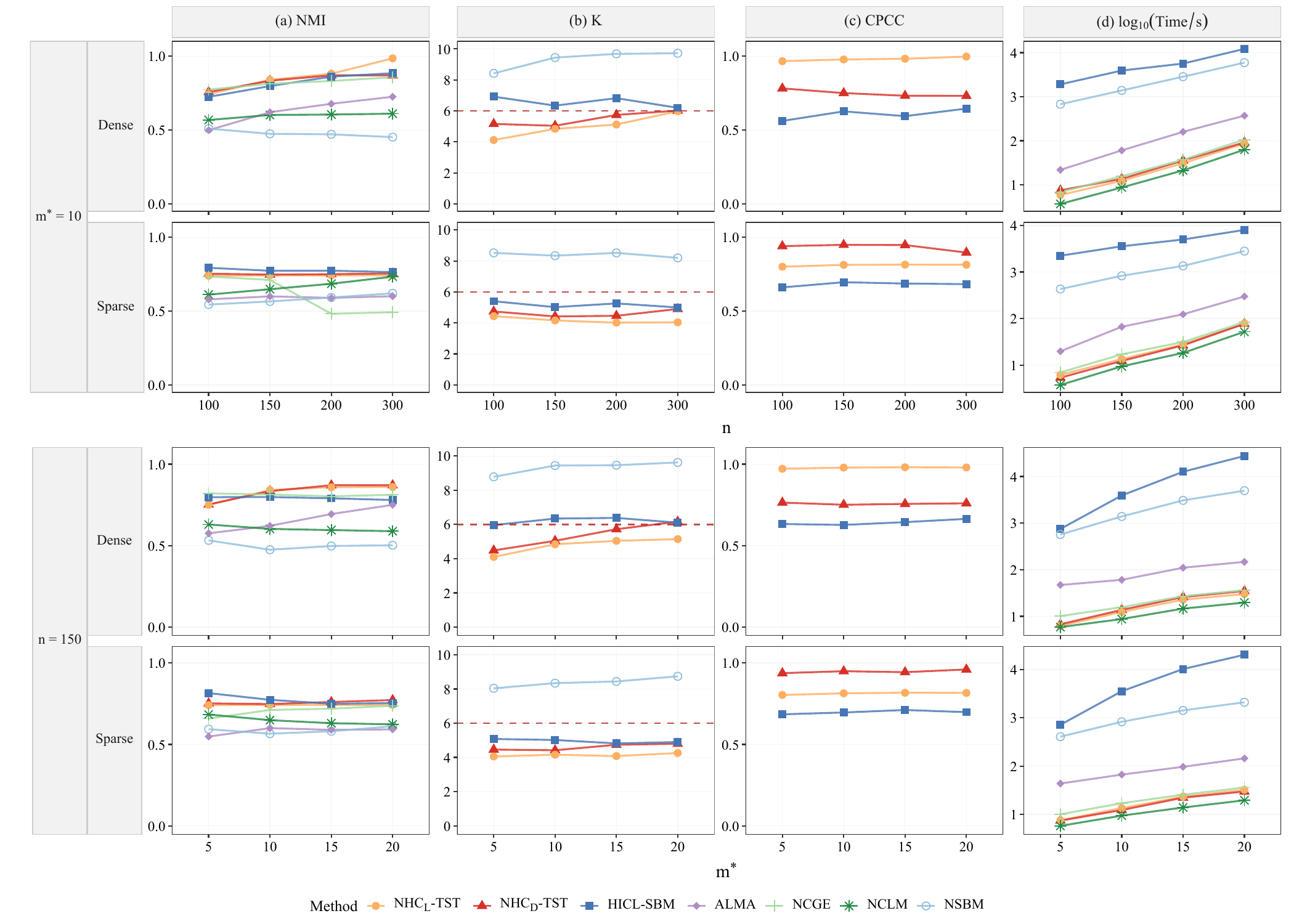}	
				\caption{Behavior comparison when the underlying networks are generated from non-uniform hierarchical smooth graphons. Panels report (a)  NMI; (b) the estimated total number of clusters; (c) CPCC for $\text{NHC}_{\text{L}}\text{-TST}$, $\text{NHC}_{\text{D}}\text{-TST}$, and HICL-SBM; and (d) log transformed running time (in seconds).}
%				\small \textbf{Alt text:} A grid of line plots evaluating various clustering methods on networks generated from non-uniform hierarchical smooth graphons. The plots depict NMI, estimated cluster number K, CPCC, and log 10 running time under dense and sparse conditions  as the number of networks or nodes number varies.
				\label{fig:trimmed-res}
			\end{figure}

			Overall, our proposed methods $\text{NHC}_{\text{D}}\text{-TST}$ and $\text{NHC}_{\text{L}}\text{-TST}$  achieve a highly favorable trade-off between accuracy,  robustness, and efficiency. As shown in Figure~\ref{fig:trimmed-res}, HICL-SBM suffers from severe computational bottlenecks, whereas both NHC-TST variants remain high efficiency, running nearly as fast as the fastest baseline NCLM.
			In dense networks, the NHC-TST framework predominantly attains the highest NMI scores. Regarding the estimated number of clusters, both NHC-TST and HICL-SBM successfully avoid the severe overestimation exhibited by NSBM; at larger data scales, $\text{NHC}_{\text{D}}\text{-TST}$ provides the most accurate estimation, converging tightly to the ground truth $K = 6$. Furthermore, the NHC-TST variants strictly outperform HICL-SBM in terms of CPCC, with the values for $\text{NHC}_{\text{L}}\text{-TST}$ approaching $1$.
			In sparse networks, although HICL-SBM holds a marginal advantage in NMI and $K$ estimation, this minor gain comes at an excessive computational price. More importantly, the NHC-TST framework demonstrates a distinct advantage in recovering the latent hierarchy despite the mild HDM violation. $\text{NHC}_{\text{D}}\text{-TST}$ maintains optimal CPCC scores near $1$ while HICL-SBM systematically underperforms.

			\section{Application to migration networks}
			\label{sec:application}
			
			%In this section, we apply our method to a novel dataset of monthly country-to-country migration flow estimates from January 2019 to December 2022 \citep{chi2025measuring}.
			%This dataset is constructed from privacy-protected Facebook user data and offers high-resolution estimations of global bilateral migration dynamics. While the original study includes 181 countries, the publicly available dataset provides records for 180 countries. We formulate the global migration system as a sequence of temporal networks, where the 180 countries serve as nodes and the total monthly number of migrants between any two countries defines the edge weights, resulting in 48 directed, weighted networks that capture global monthly migration flows over the four-year period.
			In this section, we apply our method to the global migration dataset introduced by \cite{chi2025measuring}. While the original study includes 181 countries, the publicly available dataset provides records for 180 countries. We formulate the global migration system as a sequence of temporal networks, where the 180 countries serve as nodes and the total monthly number of migrants between any two countries defines the edge weights, resulting in 48 monthly directed, weighted networks spanning January 2019 to December 2022.
			
			To capture core structural relationships between countries and align with the clustering approaches, we preprocess the initial networks into a sequence of undirected binary networks. First, we convert directed flows to undirected interactions by summing bidirectional migration volumes between each country pair, reflecting total bilateral migration intensity. We then apply a data-driven thresholding strategy to filter topological noise: for each monthly network, we binarize edges by retaining only those whose migration volumes exceed the $\tau$-th quantile of all non-zero flows.
			%we retain only the edges whose migration volumes exceed the $\tau$-th quantile of all non-zero flows. Finally, we assign a value of 1 to the remaining edges. 
			
			Our goal is to cluster the 48 networks and evaluate whether they can be meaningfully grouped. Specifically, for ALMA, we tune the number of clusters $K \in \{2, \ldots, 6\}$ and the number of underlying communities per cluster $Q \in \{2, 3, 4\}$ via grid search using the MSE-based elbow method.
			To achieve this, we implement the clustering methods introduced in the simulation study and adopt the corresponding parameter selection procedures for each method.
			%we search over a candidate grid for the number of clusters $K_0 = \{2, \ldots, 6\}$ and the number of underlying communities within each cluster $Q_0 = \{2, 3, 4\}$. We select the optimal values using MSE-based elbow method.
			%The MSE-based elbow method yields $K = 5$ and $Q_0 = 3$ as the optimal parameters.  
			For NCGE and NCLM, the optimal number of clusters is analogously selected via the eigenvalue elbow heuristic. All other algorithmic settings remain identical to those specified in Section \ref{sec:sim}.
			
			We examined  13 threshold levels, corresponding to $\tau \in \{0.2, 0.25, \ldots, 0.8\}$. Across these thresholds, the proposed NHC-TST method exhibits relatively stable cluster number estimates, whereas some competing methods either collapse to a single cluster or produce highly fragmented partitions. As a representative example, we report the results for $\tau=0.7$ in Table \ref{tab:cluster_counts}. This threshold yields a network of moderate sparsity and the estimated cluster structure remains insensitive to small perturbations of the threshold.
			%Table \ref{tab:cluster_counts} summarizes the number of clusters obtained by all methods. It shows that the partition complexity varies significantly  across methods.
			Table \ref{tab:cluster_counts} shows that  $\text{NHC}_{\text{L}}\text{-TST}$/$\text{NHC}_{\text{D}}\text{-TST}$  and  ALMA obtain moderate partition number, ranging from 3 to 5. In contrast, NCGE collapses all countries into a single cluster, while the other methods yield overly fragmented partitions that are difficult to interpret at a global scale.
			
			\begin{table}[htbp]
				\caption{Number of clusters identified by different network clustering methods.}
%				\small \textbf{Alt text:} A table showing the number of clusters identified by seven different network clustering methods. 
				\centering
				\begin{tabular}{@{} l ccccccc @{}}
					\toprule
					Method	& $\text{NHC}_{\text{L}}\text{-TST}$ & $\text{NHC}_{\text{D}}\text{-TST}$ & HICL\_SBM & ALMA & NCLM & NCGE & NSBM  \\
					\midrule
					Number of clusters  & 3  & 4  & 7  & 5 & 6  & 1 & 10 \\
					\bottomrule
				\end{tabular}
				\label{tab:cluster_counts}
			\end{table}
			
			The resulting cluster assignments are presented on the bottom panel of Figure~\ref{fig: compare}. Following the report in \citet{chi2025measuring}, global migration flows during this period were heavily influenced by the profound shocks of the COVID-19 pandemic and the Russia-Ukraine war. To better interpret the clustering results, we display the COVID-19 policy stringency index alongside four pivotal chronological markers provided in \citet{chi2025measuring} on the top panel: the pandemic's onset, the rebound phase, the restoration of pre-pandemic levels, and the outbreak of the war.
			
			\begin{figure}[!htb]
				\centering
				\includegraphics[width=1\linewidth]{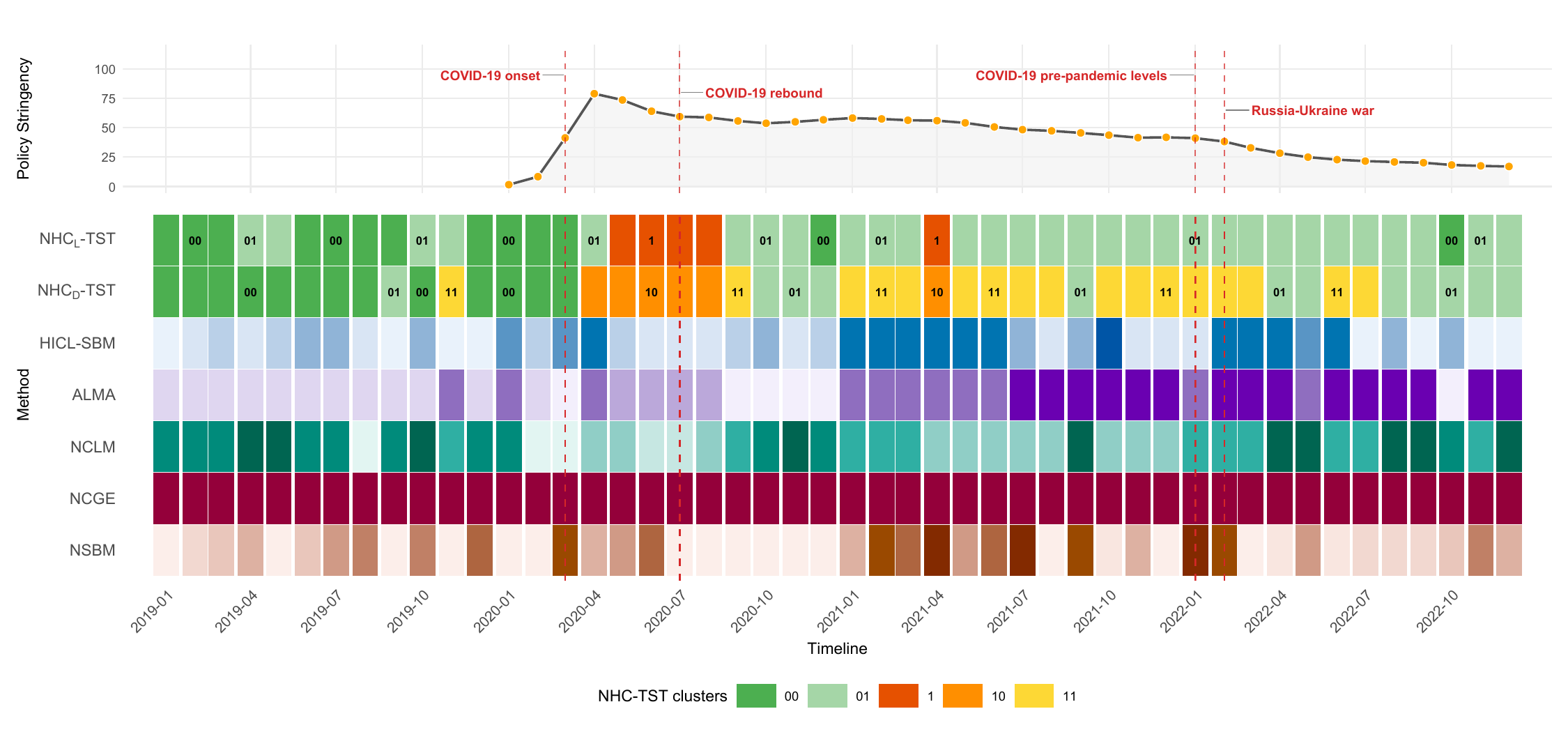}	
				\caption{Cluster assignments on migration networks. Top: Policy stringency. Bottom: Assignments. $\text{NHC}_{\text{L}}\text{-TST}$ and $\text{NHC}_{\text{D}}\text{-TST}$ share a unified palette with on-tile labels denoting hierarchical states, while other flat methods adopt distinct, unlegended palettes to differentiate categorical clusters.}
%				\small	\textbf{Alt text:} A two-part visual panel chart of migration networks over the four-year period. The upper panel displays policy stringency levels. The lower panels compare cluster assignments across different algorithms, using a shared color palette with text labels for the NHC-TST method and varying unlegended color palettes for other flat clustering methods.
				\label{fig: compare}
			\end{figure}	
			
			Within our framework, $\text{NHC}_{\text{L}}\text{-TST}$ partitions the timeline into three  clusters. It clearly isolates the period between the COVID-19 onset and the rebound as a unified crisis state (cluster 1). Meanwhile, the periods both before the onset and after the rebound are predominantly grouped into clusters 00 and 01. This indicates that once  the most severe phase of the pandemic passed, the network reverted to a structure similar to its pre-shock state.
			
			Method $\text{NHC}_{\text{D}}\text{-TST}$ captures finer-grained event shocks by partitioning the timeline into four distinct clusters, which align with the four chronological markers.
			Branch 1 isolates the active crisis, grouping the pandemic onset (10) and the subsequent rebound (11).
			Meanwhile, Branch 0 connects the pre-pandemic phase (00) and recovery period (01).
			Notably, the 2022 recovery is classified as cluster 01 rather than 00. % a direct return to the original 00 state. 
			This aligns with historical events: while the pandemic stabilized by January 2022, the immediate outbreak of the Russia-Ukraine war in February 2022 triggered a new geopolitical conflict.
			
			These structural insights from $\text{NHC}_{\text{D}}\text{-TST}$ are further supported by the pairwise interaction patterns shown in Figure~\ref{fig: chordal}.
			To capture the key features of each temporal state, we construct the diagrams based on anomalous interaction intensities rather than absolute migration volumes. For each cluster, we compute its average adjacency matrix and subtract the global average matrix over the  4-year period. Negative values are set to 0, thereby isolating the connections that exhibit a significant positive deviation from the global baseline. We then retain the top 20 edges with the largest anomalous weights to construct undirected networks.
			This procedure removes baseline effects and highlights connections that become exceptionally active in each phase.
			
			\begin{figure}[!htb]
				\centering
				\includegraphics[width=1\linewidth]{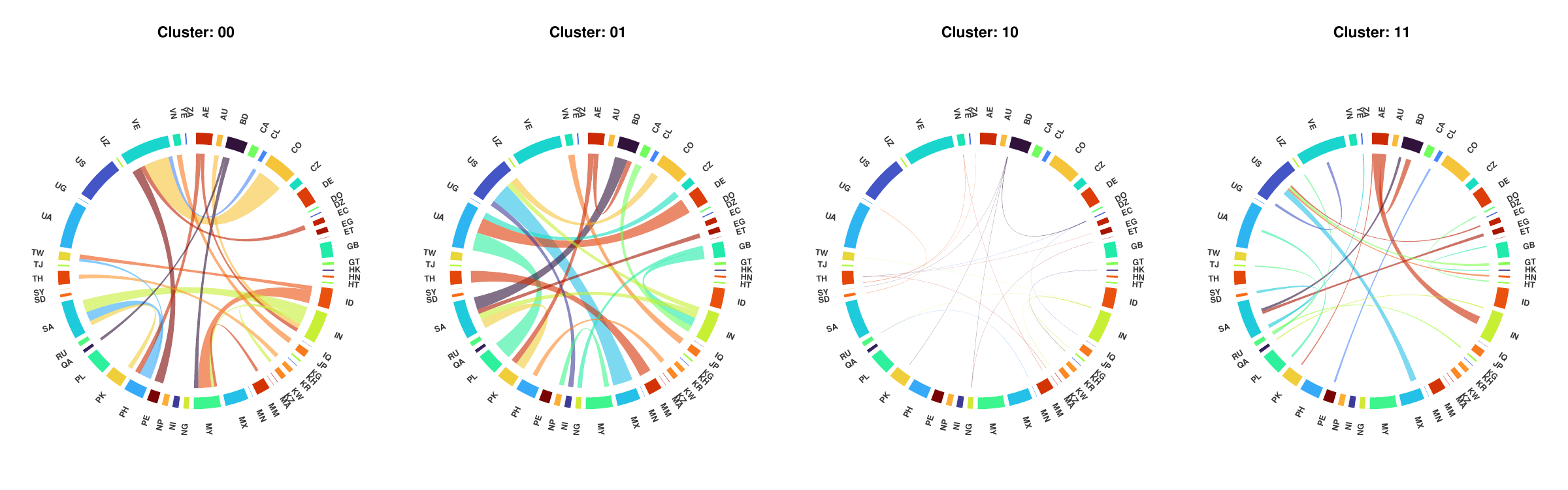}		
				\caption{Pairwise interaction networks from $\text{NHC}_{\text{D}}\text{-TST}$ based on anomalous interaction intensities, highlighting the top 20 anomalous edges and associated countries.}
%				\small \textbf{Alt text:}  Four circular chord diagrams displaying pairwise interaction networks for the four clusters obtained by the NHC-D. Within each cluster's diagram, inner connecting lines highlight the top 20 anomalous edges, while the outer ring labels the corresponding countries connected by all these edges.
				\label{fig: chordal}
			\end{figure}	
			
			As illustrated in Figure~\ref{fig: chordal}, method $\text{NHC}_{\text{D}}\text{-TST}$ effectively uncovers the temporal evolution of global migration by separating stable structural regimes (Branch 0) from severe systemic shocks (Branch 1).
			Branch 0 represents the normal baseline of global migration. Both clusters 00 and 01 share a high connection intensity driven by major labor flows between South Asian (e.g., India (IN), Pakistan (PK), and Bangladesh (BD)) and West Asian (e.g., Saudi Arabia (SA), Qatar (QA), and the United Arab Emirates (AE)).
			Within this stable branch, cluster 00 specifically highlights steady regional movements, especially the intense connection between Venezuela (VE) and Colombia (CO).
			Meanwhile, cluster 01 is heavily driven by the massive and steady migration connection between Mexico (MX) and the United States (US). %Built on this strong foundation, 
			Cluster 01 also clearly  captures the large-scale migration triggered by the Russian invasion of Ukraine, seen through the sudden emergence of migration corridors between Ukraine (UA) and European nations such as Poland (PL), Germany (DE), and the Czech Republic (CZ).  
			In contrast, Branch 1 reveals the profound systemic shock of the COVID-19 pandemic. 
			Cluster 10 characterizes the early crisis stage with a visibly sparse network, reflecting the severe drop in global migration  due to widespread lockdowns and border closures.
			Cluster 11 illustrates the rebound phase. During this stage, global flows began to recover and the network's overall connectivity gradually restored. These findings are highly consistent with the  results reported by \citet{chi2025measuring}.
			
			As for ALMA, although it  produces a moderate five-cluster partition, the resulting timeline blocks are hard to align with major historical events, especially  in the later periods. Furthermore, this flat partition cannot fully represent structural similarities across clusters, obscuring which temporal phases are topologically more closely related.
			
			\section{Conclusion and discussion}
			\label{sec:conclusion}
			
			This paper has introduced a framework for hierarchically clustering populations of networks.
			We introduce a \emph{Hierarchical Distance Matrix} (HDM) that formalizes the geometry of a tree-organized population of networks.
			On the algorithmic side, the \emph{Network Hierarchical Clustering with Two-Sample Test} (NHC-TST) procedure learns the hierarchical structures of the network populations  by interleaving spectral bipartitioning with a graphon-based two-sample stopping rule, yielding a data-driven hierarchy that requires no \emph{a~priori} knowledge  of the number of clusters. On the statistical side, our analysis shows that the two spectral variants -- $\mathrm{NHC}_{\mathrm L}\text{-TST}$ and $\mathrm{NHC}_{\mathrm D}\text{-TST}$ -- recover the true hierarchy exactly at the population level and that $\mathrm{NHC}_{\mathrm L}\text{-TST}$ is consistent in the empirical regime where link-probability matrices are estimated  via neighborhood smoothing. The simulations and the global migration application confirm that this combination of structural insight and statistical adaptivity translates into state-of-the-art practical performance.
			
			Some future directions include  extending our framework to hierarchical structures that are beyond binary trees, accommodate weighted and edge-attributed networks and possibly model dynamic networks. 
			
			\section*{Acknowledgment}
			The work of Li Chen was supported by the Scientific and Technological Innovation Team for Qinghai-Tibetan Plateau Research in Southwest Minzu University (2024CXTD19) and the Fundamental Research Funds for the Central Universities, Southwest Minzu University (ZYN2024069). Eric Kolaczyk was supported by the Natural Sciences and Engineering Research Council of Canada (NSERC), through Grants RGPIN-2023-03566 and DGDND-2023-03566, and by the Canada Research Chairs Program. Lizhen Lin would like to acknowledge the generous support of NSF grant DMS 2503119. 
			
			\bibliographystyle{biometrika}
			\bibliography{refs}
			
%			\clearpage
			\appendix
			
			\renewcommand{\thesection}{S.\arabic{section}}
			\counterwithout{equation}{section} 
			\setcounter{equation}{0}
			\renewcommand{\theequation}{S.\arabic{equation}}
			
			\section*{Appendix}
			
			This supplement includes detailed proof for all the main results presented in the paper, along with supporting lemmas.

			\section{Preliminaries}
			
			%	The following lemma about the class of Positive-Negative block matrix is a preliminary for the proof.
			
			\begin{lemma}[Positive-Negative block matrix \citep{balakrishnan2011noise}] \label{Positive-Negative}
				Let $B$ be an $m \times m$ symmetric matrix with the Positive-Negative block structure of 
				\begin{equation*}
					B = 
					\begin{pNiceArray}{c|c}[cell-space-limits = 2pt]
						B_+ & B_- \\
						\Hline 
						B_-^T & \tilde{B}_+ \\
					\end{pNiceArray},
				\end{equation*}
				where, for $p, q \geq 1$, the $p \times p$ block $B_+$ and the $q \times q$ block $\tilde{B}_+$ have strictly positive off-diagonal entries, while the $p \times q$ block $B_-$ has strictly negative entries. Let $v$ be the dominant eigenvector of $B$. Then the eigenvector $v$ either has the sign pattern of $\begin{pmatrix} v_+ \\ v_- \end{pmatrix}$, where $v_+$, the first $p$ elements of $v$, are strictly positive and $v_-$, other $q$ elements of $v$, are strictly negative or has the reverse sign pattern.
			\end{lemma}
			
			%	Before we proceed to the proof, we state a result in Spectral Graph Theory that we will use:
			\begin{lemma}[Lemma 6.3.1 in \citet{Spielman2009}, Corollary 4.3.12 in \citet{horn2013matrix}] \label{lemma spe}
				If $G$ and $H$ are graphs whose Laplacian matrices satisfy 
				\[
				L_G \succcurlyeq c\,L_H,
				\]
				for some $c > 0$. That is, $L_G - c\,L_H$ is positive semidefinite, then for each $i$,
				\[
				\lambda_i(G) \;\ge\; c\,\lambda_i(H).
				\]		
			\end{lemma}
			
			%		We now introduce an eigenvalue stability inequality that we will use later (see Corollary III.2.6 in \citet{bhatia2013matrix}).
			
			\begin{lemma}[Weyl's Inequality\citep{bhatia2013matrix}]\label{weyl}
				Let $A$ and $B$ be two $m \times m$ Hermitian matrices, we have the following eigenvalue stability inequality
				\begin{equation*}
					\max_{i = 1, 2, \ldots, m} |\lambda_i(A) - \lambda_i(B)| \leq \|A - B\|_{\mathrm{op}} \enskip,
				\end{equation*}
				where $\|\cdot\|_{\mathrm{op}}$  denotes the operator norm.
			\end{lemma}
			
			%	\textbf{Graphon estimation error.}
			%	Under the assumption that an NBS-type estimator is used for each link probability matrix, the estimation error satisfies the following bound \citep{zhao2019change}.
			%	\begin{equation}\label{NBS}
				%		\|\hat P^{(i)} - P^{(i)} \|_F \leq C_0 \sqrt{n \log n} \quad \text{for all} \ i = 1, 2, \ldots, m,
				%	\end{equation}
			%	with probability at least $1 - n^{- \varepsilon}$ for any $\varepsilon > 0$, where $C_0$ is a positive global constant depending on $\varepsilon$ and another global constant $B_0 > 0$ but not on $n$ or $m$.
			
			\section{Proof of Theorem \ref{theorem L_D}}
			%	\begin{theorem}\label{theorem L_D}		
				%		Let $D \in \mathbb{R}^{m \times m}$ be a normalized population distance matrix, and $L_D = \text{diag}(D \mathbf{1}) - D$ be its associated Laplacian matrix. Suppose $D$ is an HDM. Under Assumptions \ref{A1} and \ref{A2}, given an oracle stopping rule, the $\text{NHC}_{\text{L}}$ method will exactly recover the true hierarchical topology.
				%	\end{theorem}
			
			%	\begin{proof}
				Consider the first split of the hierarchical clustering and the corresponding distance matrix is block-partitioned according to \eqref{eq:D}.		
				Since $\beta_0<\alpha_1$, there exists some constant $\kappa \in (\beta_0, \alpha_1)$, such that the constructed matrix $B$  below admits a Positive-Negative block matrix representation.
				\begin{equation*}
					B = \kappa \mathbf{1} \mathbf{1}^T + L_D = \kappa \mathbf{1} \mathbf{1}^T + \text{diag}(D \mathbf{1}) - D = 
					\begin{pNiceArray}{c|c}[cell-space-limits = 2pt]
						B_+ & B_- \\
						\Hline
						B_-^T & \tilde{B}_+ \\
					\end{pNiceArray}.
				\end{equation*}
				The diagonal blocks $B_+ \in \mathbb{R}^{m_0 \times m_0}, \tilde{B}_+ \in \mathbb{R}^{m_1 \times m_1}$, and the off-diagonal block $B_- \in \mathbb{R}^{m_0 \times m_1}$.

				Let $\{\lambda_i(B)\}_{i=1}^m$ and $\{u_i\}_{i=1}^m$ be the eigenvalues and corresponding eigenvectors of $B$, and let $\{\lambda_i(L_D)\}_{i=1}^m$ and $\{v_i\}_{i=1}^m$ be those of $L_D$.
				Applying Lemma \ref{Positive-Negative}, we know that the dominant eigenvector $u_1$ has the sign pattern of $\begin{pmatrix} u_{1+} \\ u_{1-} \end{pmatrix}$. Part $u_{1+}$ are the first $m_0$ elements and are strictly positive, part $u_{1-}$ are the other $m_1$ elements and are strictly negative, or has the reverse sign pattern.
				
				Now we prove that $u_1$ also is the dominant eigenvector of $L_D$.	
				
				For $B$, we have
				\begin{equation*}
					B \mathbf{1} = (\kappa \mathbf{1} \mathbf{1}^T + L_D) \mathbf{1} = \kappa m \mathbf{1}.
				\end{equation*}
				The last equation holds true because as a Laplacian matrix,  $L_D$ has an eigenvector $\mathbf{1}$ with eigenvalue 0.
				So $B$ has an eigenvector $\mathbf{1}$ with eigenvalue $\kappa m$. Observe that the signs of  $\mathbf{1}$ are the same, implying $u_1 \neq \mathbf{1}$. Consequently, $u_1$ must belong to the set of remaining eigenvectors orthogonal to $\mathbf{1}$, which we denote as $S_B \triangleq \{u_i \mid B u_i = \lambda_i(B) u_i, u_i^\top \mathbf{1} = 0\}$.
				
				Recall that vector $\mathbf{1}$ corresponds to the zero eigenvalue of $L_D$. 	
				By definition, $L_D = \text{diag}(D \mathbf{1}) - D$. Because a splittable distance matrix $D$ contains positive off-diagonal entries, ensuring that the trace of $L_D$ is strictly positive. 
				This inherently forces the dominant eigenvalue of the positive semi-definite matrix $L_D$ to be strictly positive.
				Consequently, $v_1$ must belong to the set of remaining eigenvectors orthogonal to $\mathbf{1}$, which we denote as $S_L \triangleq \{v_i \mid L_D v_i = \lambda_i(L_D) v_i, v_i^\top \mathbf{1} = 0\}$.
				
				Restricting our analysis to the eigenspaces orthogonal to $\mathbf{1}$, we consider $S_B$ and $S_L$.
				
				If $u_i \in S_B$, we have $\mathbf{1}^T u_{i} = 0$, and
				\begin{equation*}
					L_D u_i = (B - \kappa \mathbf{11}^T) u_i = B u_i = \lambda_i(B) u_i. 
				\end{equation*}
				So $u_i$ is also an eigenvector of $L_D$ with $B$'s eigenvalue $\lambda_i(B)$, i.e., $S_B \subseteq S_L$. 
				
				If $v_i \in S_L$, we have $\mathbf{1}^T v_{i} = 0$, and
				\begin{equation*}
					B v_i = (\kappa \mathbf{11}^T + L_D) v_i = L_D v_i = \lambda_i(L_D) v_i.
				\end{equation*}
				So $v_i$ is also an eigenvector of $B$ with $L_D$'s eigenvalue $\lambda_i(L_D)$, i.e., $S_L \subseteq S_B$. 
				
				Consequently, we have $S_B = S_L$, with identical corresponding eigenvalues throughout the shared spectrum.
				Therefore, $u_1$ is also the dominant eigenvector of $L_D$. The structure of  $v_1$ exhibits the sign pattern of $u_1$: the first $m_0$ components are strictly positive and the last $m_1$ components are strictly negative (or vice versa). So $\text{NHC}_{\text{L}}$ method  naturally exploits this sign separation to correctly recover the first split. Repeated application of this result concludes the proof.	
				%	\end{proof}

			\section{Proof of Theorem \ref{theorem D}}
			%	\begin{theorem}\label{theorem D}
				%		%	For a normalized distance matrix $D \in \mathbb{R}^{m \times m}$ satisfying the monotonicity property described in Definition \ref{def:HDM}, let 
				%		Let $D \in \mathbb{R}^{m \times m}$ be a normalized population distance matrix and let
				%		\begin{equation}
					%			%g(m_{0-},m_{1+}) = \frac{- \beta_0(m_{0-}^2 + m_{1+}^2) + 2 m_{1+} m_{0-} \beta_0+ \beta_0(m_{0-} - m_{1+})(m_0 - m_1)}{m_0 m_{1+} + m_1 m_{0-} - 2 m_{1+} m_{0-}},
					%			g(m_{0-}, m_{0+}, m_{1-}, m_{1+}) =  \frac{\beta_0 (m_{0-} - m_{1+}) (m_{0+} - m_{1-})}{ m_{0+} m_{1+} + m_{1-} m_{0-}},
					%		\end{equation}
				%		where $m_{0+} + m_{0-} = m_0, m_{1+} + m_{1-} = m_1, 1 \leq m_{0-} \leq m_0 - 1$, $1 \leq m_{1+} \leq m_1 - 1$. Define $M = \max\limits_{m_{0-}, m_{0+}, m_{1-}, m_{1+}} g(m_{0-}, m_{0+}, m_{1-}, m_{1+})$. 
				%		Suppose $D$ is an HDM and Assumptions \ref{A1} -- \ref{A3} hold.
				%		Given an oracle stopping rule, method $\text{NHC}_{\text{D}}$ will exactly recover the initial root bipartition if the following separation conditions hold:
				%		$$
				%		\text{(1)} \quad \alpha_1 > \beta_0 + M, \qquad \text{and} \qquad \text{(2)} \quad \alpha_1 > \max\{\eta, \eta^{-1}\} \beta_0.
				%		$$
				%		%\begin{enumerate}[label=(\arabic*)]
				%		%	\item $\alpha_1 > \beta_0 + M \enskip.$ 
				%		%	\item $\alpha_1 > \max\left\{\eta, 1 / \eta \right\} \beta_0 \enskip.$
				%		%\end{enumerate}
				%		
				%		Furthermore, provided that analogous local separation conditions are satisfied at all internal nodes, the recursive application of $\text{NHC}_{\text{D}}$ will exactly recover the true hierarchical topology.
				%	\end{theorem}
			
			%\begin{proof}
			Consider the first split of the hierarchical clustering, where the population distance matrix $D$ takes the form given in \eqref{eq:D}.
			Let $u = (u_+, u_-)^\top$ be the  minimal unit eigenvector of $D$ where $u_+$ are the first $m_0$ elements and $u_-$ are the other $m_1$ elements. The signs of its entries are undetermined. 
			
			We will prove that $u_+$ and $u_-$ have opposite signs. The proof proceeds by considering two main cases depending on whether $u$ contains zero elements.
			
			\textbf{Case 1: Suppose $u_i \neq 0$ for all $i = 1, 2, \ldots, m$.} 
			
			Let $I_+, I_-$ be index sets of positive and negative elements in $u_+$, with their cardinalities $|I_+| = m_{0+}, |I_-| = m_{0-}, m_{0+} + m_{0-} = m_0$. Similarly, let $J_+$ and $J_-$ represent the corresponding index sets for $u_-$, with sizes $|J_+| = m_{1+}, |J_-| = m_{1-}, m_{1+} + m_{1-} = m_1$. Then
			\begin{align}
				u^\top D u 
				&= u_{+}^\top D^{[0]} u_+ + u_{-}^\top D^{[1]} u_- + 2 u_{+}^\top R u_- \nonumber\\
				&= \Bigg( \sum_{i, j \in I_+} u_i D_{ij} u_j + \underbrace{2 \sum_{i \in I_+, j \in I_-} u_i D_{ij} u_j}_{\mathcal{T}_1} + \sum_{i, j \in I_-} u_i D_{ij} u_j \Bigg) \nonumber\\
				&\quad + \Bigg( \sum_{i, j \in J_+} u_i D_{ij} u_j + \underbrace{2 \sum_{i \in J_+, j \in J_-} u_i D_{ij} u_j}_{\mathcal{T}_2} + \sum_{i, j \in J_-} u_i D_{ij} u_j \Bigg) \nonumber\\
				&\quad + 2 \Bigg( \underbrace{\sum_{i \in I_+, j \in J_+} u_i D_{ij} u_j}_{\mathcal{T}_3} + \sum_{i \in I_+, j \in J_-} u_i D_{ij} u_j 
				+ \sum_{i \in I_-, j \in J_+} u_i D_{ij} u_j + \underbrace{\sum_{i \in I_-, j \in J_-} u_i D_{ij} u_j}_{\mathcal{T}_4} \Bigg) \enskip. \label{eqt}
			\end{align}
			
			By the Rayleigh-Ritz theorem (Section 5.5.2 of \citet{lutkepohl1997handbook}), vector $u$ would be the minimal eigenvector when  $u^\top D u$ reach the minimum value (subject to $\|u\| = 1$). 
			Without loss of generality, we will analyze the mixed sign patterns across subsets $I_+, I_-, J_+$, and $J_-$ from two representative cases.
			
			\textbf{(i) All the subsets $I_+, I_-, J_+$, and $J_-$ are strictly non-empty.} 
			
			For this fully mixed case, $m_{0+}, m_{0-}, m_{1+}, m_{1-}$ are all at least 1. To analyze when $u^\top D u$ reaches the minimum value, let us construct a new vector $\tilde{u}$ by flipping the signs of all elements in $I_-$ and $J_+$. We now evaluate the value change: $$\tilde{u}^\top D \tilde{u} - u^\top D u \enskip.$$
			Notice that flipping these signs only alters the cross-terms between flipped and unflipped index sets. Specifically, in equation \eqref{eqt}, terms $\mathcal{T}_1$ and $\mathcal{T}_2$ will strictly increase (transitioning from negative to positive products), whereas terms $\mathcal{T}_3$ and $\mathcal{T}_4$ will strictly decrease. Thus, the value change can be bounded as:	
			\begin{align} \label{diff}
				& \quad \ \tilde{u}^T D \tilde{u} - u^T D u  \nonumber\\
				& =  - 4 \sum_{i \in I_+, j \in I_-} u_i D_{ij} u_j - 4 \sum_{i \in J_+, j \in J_-} u_i D_{ij} u_j - 4 \sum_{i \in I_+, j \in J_+} u_i D_{ij} u_j - 4 \sum_{i \in I_-, j \in J_-} u_i D_{ij} u_j \nonumber \\
				& \leq  4 \beta_0(m_{0+} m_{0-} + m_{1+} m_{1-}) - 4 \alpha_1(m_{0+} m_{1+} + m_{0-} m_{1-}) \nonumber\\
				& = 4 \beta_0 \big((m_0 - m_{0-}) m_{0-} + (m_1 - m_{1+}) m_{1+}\big) - 
				4 \alpha_1 \big( (m_0 - m_{0-}) m_{1+} + (m_1 - m_{1+}) m_{0-} \big) \enskip.
			\end{align}
			
			Denoting $\delta = \alpha_1 - \beta_0$, we insert $\alpha_1  = \delta + \beta_0$ to \eqref{diff} and  obtain:
			\begin{align*}
				& \quad \ \tilde{u}^T D \tilde{u} - u^T D u  \nonumber\\
				& \leq 4 \big(2 m_{1+} m_{0-} (\beta_0 + \delta) - \beta_0(m_{0-}^2 + m_{1+}^2) + \beta_0 (m_{0-} - m_{1+})(m_0 - m_1) - \delta (m_0 m_{1+} + m_1 m_{0-})\big) \enskip.
			\end{align*}
			To ensure the objective function $u^\top D u $ strictly decreases (i.e., $\tilde{u}^T D \tilde{u} - u^T D u < 0$), the signal gap $\delta$ must satisfy 
			\begin{align*}
				\delta & > \frac{- \beta_0(m_{0-}^2 + m_{1+}^2) + 2 m_{1+} m_{0-} \beta_0+ \beta_0(m_{0-} - m_{1+})(m_0 - m_1)}{m_0 m_{1+} + m_1 m_{0-} - 2 m_{1+} m_{0-}} \nonumber \\
				& = \frac{- \beta_0(m_{0-} - m_{1+})^2 + \beta_0(m_{0-} - m_{1+})(m_0 - m_1)}{(m_{0-} + m_{0+})m_{1+} + (m_{1-} + m_{1+})m_{0-} - 2 m_{1+} m_{0-}} \\
				& = \frac{\beta_0 (m_{0-} - m_{1+}) (m_{0+} - m_{1-})}{ m_{0+} m_{1+} + m_{1-} m_{0-}} = g(m_{0-}, m_{0+}, m_{1-}, m_{1+}) \enskip .
			\end{align*}
			
			According to the regimes of $m_{0-}, m_{0+}, m_{1-}, m_{1+}$, the denominator is strictly positive. By taking the maximum over all values, we can guarantee that $M = \max\limits_{m_{0-}, m_{0+}, m_{1-}, m_{1+}} g(m_{0-}, m_{0+}, m_{1-}, m_{1+}) \geq 0$. 	
			When $\delta > M$, i.e., $$\alpha_1 > \beta_0 + M,$$
			we have  $\tilde{u}^T D \tilde{u} - u^T D u < 0$. 
			Consequently, minimizing $u^T D u$ among all unit-norm vectors forces the sets $I_-$ and $J_+$ to be empty, thereby ensuring that the elements of $u_+$ and $u_-$ have uniform, yet strictly opposite, signs.
			
			\textbf{(ii) Subsets $I_+$ and $I_-$ are strictly non-empty, but $J_+$ is empty (i.e., $m_{0+} \geq 1 , m_{0-} \ge 1, m_{1+} = 0, m_{1-} = m_1 $).} 
			
			Similarly, we construct a new vector $\tilde{u}$ by flipping the signs of all elements in $I_-$. Under this construction, the term $\mathcal{T}_1$ in equation \eqref{eqt} strictly increases, whereas $\mathcal{T}_4$ strictly decreases. Thus, the value change can be bounded as:	
			\begin{equation*}
				\tilde{u}^\top D \tilde{u} - u^\top D u \leq 4 m_{0+} m_{0-} \beta_0 - 4 m_{0-} m_1 \alpha_1 \enskip.
			\end{equation*}
			To guarantee that the difference is strictly negative (i.e., $\tilde{u}^\top D \tilde{u} - u^\top D u < 0$), it suffices to require
			\begin{equation*}
				\alpha_1 > \frac{m_{0+}}{m_1} \beta_0 \enskip. 
			\end{equation*}
			Consequently, minimizing $u^\top D u$ universally forces $I_-$ to be empty.
			
			Analogously, if we consider the symmetric case where $I_-$ is empty while $J_+$ and $J_-$ are strictly non-empty, we obtain 
			\begin{equation*}
				\alpha_1 > \frac{m_{1-}}{m_0} \beta_0 \enskip.
			\end{equation*}
			Similarly, in this scenario,  minimizing $u^\top D u$ universally forces $J_+$ to be empty.
			
			By considering the worst-case upper bounds (where $m_{0+}$ is bounded by $m_0$ and $m_{1-}$ is bounded by $m_1$), we naturally arrive at the global sufficient condition:
			\begin{equation*}
				\alpha_1 > \max\left\{\frac{m_{0}}{m_1}, \frac{m_{1}}{m_0}\right\} \beta_0 = \max\left\{\eta, \frac{1}{\eta}\right\} \beta_0 \enskip.
			\end{equation*}
			Under this condition, the elements of $u_+$ and $u_-$ are guaranteed to have uniform, yet strictly opposite, signs.
			
			\textbf{Case 2: Suppose $u_{r} = 0$ for some index $r \in \{1, 2, \ldots, m \}$.} 
			
			Let $D_{r\cdot}$ be the $r$-th row of $D$. From the eigenvalue equation, we have 
			\begin{equation} \label{eigen equ}
				D_{r \cdot} u = \sum_{j \neq r} D_{r j} u_j = \lambda_m(D) u_r = 0.			
			\end{equation}		
			Let us form a new vector $\hat u$ by changing $u_r$ to any nonzero value $\hat u_r \neq 0$:
			\begin{equation*}
				\hat u_i = \begin{cases}
					u_i, & i \neq r, \\
					\hat u_r, &i = r.
				\end{cases}
			\end{equation*}
			To obtain the value change of $u^T D u$, we first calculate $(\hat u)^T D \hat u$,
			\begin{equation*}\label{equ 2}
				(\hat u)^T D \hat u = \sum_{\substack{i \neq r \\ j \neq r}}  u_i u_j D_{ ij} + 
				2 \hat u_r \sum_{j \neq r} D_{r j} u_j + \hat u_r^2 D_{rr} = \sum_{\substack{i \neq r \\ j \neq r}}  u_i u_j D_{ ij} =  u^T D u \enskip.
			\end{equation*}
			The second equation comes from \eqref{eigen equ} and that $D_{rr} = 0$. By the Rayleigh-Ritz theorem (Section 5.5.2 of \citet{lutkepohl1997handbook}), we have 
			\begin{equation*}	
				\lambda_m(D) = \min_{\substack{v \in \mathbb{R}^m \\ v \neq \textbf{0}}} \frac{v^T D v}{v^T v} \leq \frac{(\hat u)^T D \hat u}{(\hat u)^T \hat u} = \frac{u^T D u}{\sum_{i \neq r} u_i^2 + (\hat u_r)^2} < \frac{u^T D u}{u^T u} = u^T D u \enskip .
			\end{equation*}
			So $u^T D u$ is not the smallest and this contradicts the assumption that $u$ is the minimal eigenvector.
			
			Consequently, all components of $u$ must be nonzero. 
			Combining the results from both Case 1 and Case 2, we conclude that the subvector  $u_+$ is strictly positive while $u_-$ is strictly negative (or vice versa), 
			revealing a clean sign pattern in $u$. Method $\text{NHC}_{\text{D}}$ thus exploits this sign separation to correctly recover the first split. %Repeated application of this result concludes the proof.	
			%\end{proof}
			
			\section{Proof of Theorem \ref{theorem sign}}
			%	Inspired by \citet{balakrishnan2011noise}, our investigation begins with a simplified class of matrices possessing the identical topological structure of a HDM, yet with constant off-diagonal distances.
			We first investigate a simplified class of matrices that share the exact structure of an HM, but with constant off-diagonal entries. Building on this construction, we then state and prove four key lemmas, which serve as the primary steps  for establishing Theorem \ref{theorem sign}. 
			
			\begin{definition}[Hierarchical Constant Matrix]\label{def:HCM}
				A matrix $H$ is a \emph{Hierarchical Constant Matrix (HCM)} if it shares the identical hierarchical cluster structure and block-partitioning scheme as HM $D$ described in Definition \ref{def:HM}, but possesses constant off-diagonal blocks. Specifically, for any internal node $x$ in the hierarchy, the between-cluster distance block $R^{[x]}$ is a constant matrix, denoted as $B^{[x]} = b^x \mathbf{1}\mathbf{1}^\top$, where $b^x > 0$ is a scalar.
				%	
				%	\hfill $\square$
			\end{definition}
			
			For instance, expanding $H$ down to the second hierarchical layer yields the following block structure:
			\begin{equation}\label{H_layer2}
				H = \begin{pNiceArray}{cc|cc}[cell-space-limits = 3pt, margin=1pt]
					H^{[00]} & B^{[0]} & \Block[borders=left]{2-2}{B} & \\
					\big(B^{[0]}\big)^\top & H^{[01]} & & \\	
					\Hline
					\Block[borders=right]{2-2} {B^\top} & & H^{[10]} & B^{[1]}  \\
					& & \big(B^{[1]}\big)^\top & H^{[11]} 
				\end{pNiceArray} \enskip,
			\end{equation}
			where $B, B^{[0]}$, and $B^{[1]}$ are the constant between-blocks with identical entries equal to the scalars $b, b^0$, and $b^1$, respectively. By definition, the diagonal blocks $H^{[00]}, H^{[01]}, H^{[10]}$, and $H^{[11]}$ recursively inherit the HCM property or reduce to scalar zeros at the terminal leaf level.
			
			Since the HM $D$ satisfies the monotonicity property along the binary tree, the same monotonic structure is inherited by the corresponding HCM representation $H$, implying that
			\begin{equation*}
				\max\{b^{x0},\, b^{x1}\} \;<\; b^{x} .
			\end{equation*}
			
			By construction, the HCM $H$ is a special case of the HM $D$. It strictly satisfies assumptions \ref{A1}--\ref{A3}, inheriting the identical within-cluster bounds and balance parameters.  
			%Notably, at each cluster index $x$, the between-cluster parameters in assumption \ref{A1} reduce to a single value:
			%\[
			%\alpha_1^{x} = \beta_1^{x} = b^{x} .
			%\]
			%We adopt this parameterization for $H$ throughout the remainder of the analysis.
			
			\begin{lemma}[Spectrum of HCMs]\label{lemma HCM}
				Consider an HCM $H \in \mathbb{R}^{m \times m}$ defined in Definition \ref{def:HCM}.
				%characterized by a collection of constant cross-cluster distances $\{b^x\}$ at each internal node $x$. 
				Suppose that $H$ satisfies assumptions \ref{A1}--\ref{A3} with  parameters $(m_x, \alpha_0^x, \beta_0^x, \alpha_1^x, \beta_1^x, \eta^x)$. 
				Then,
				\begin{enumerate}[label=(\arabic*)]
					\item $\lambda_1(L_H) = m b$.
					\item $\lambda_2(L_H) = \max \left\{ m_0 b^0 + m_1 b, \, m_1 b^1 + m_0 b \right\}$.
				\end{enumerate}
			\end{lemma}
			
			\begin{proof}
				\textbf{(i) Proof of claim (1).}
				
				We complete the proof through a two-phase argument. We first construct an unit eigenvector $\tilde{u}$ of $L_H$ with eigenvalue $m b$. Through mathematical induction, we then demonstrate that $\tilde{u}$ is the dominant eigenvector $u$, thereby establishing the desired results.
				
				Let \begin{equation*}
					\tilde{u} = \Big(\underbrace {\sqrt{\frac{m_1}{m m_0}},\ldots, \sqrt{\frac{m_1}{mm_0}}}_{m_0}, \underbrace{- \sqrt{\frac{m_0}{mm_1}}, \ldots,- \sqrt{\frac{m_0}{mm_1}}}_{m_1}\Big)^T.
				\end{equation*}
				
				Remember that $m_0, m_1$ are the sizes of the two clusters at the first layer. We first verify that $\tilde{u}$ is an eigenvector of $L_H$ with eigenvalue $m b$.	
				
				%When the depth of HCM $H$ is 0, we have $H = b(\mathbf{1}\mathbf{1}^\top - I)$, yielding the Laplacian $L_H =  \mathrm{diag}(H\mathbf 1) - H = m b I - b\mathbf{1}\mathbf{1}^\top$. 
				%It immediately follows that
				%\begin{equation}\label{equ 1}
				%L_H \tilde{u} = (m b I - b\mathbf{1}\mathbf{1}^\top) \tilde{u} = m b \tilde{u} - b \mathbf{1} (\mathbf{1}^\top \tilde{u}) = m b \tilde{u} - b \sum_{i = 1}^{m}\tilde{u}_i \mathbf{1} =  m b \tilde{u} \enskip.
				%\end{equation}
				%This implies that $\tilde{u}$ is an eigenvector of $L_H$ with eigenvalue of $mb$.	
				When the depth of HCM $H$ is 0,both $H$ and $L_H$ reduce to zero matrices, and $b = 0$. It immediately follows that $\tilde{u}$ is an eigenvector of $L_H$ associated with the eigenvalue of $mb = 0$.	
				
				Similarly, when the depth of HCM $H$ is no less than 1, it admits a natural block partition $$H = \begin{pmatrix} H^{[0]} & B \\ B^\top & H^{[1]} \end{pmatrix},$$ with $B = b \mathbf{1}_{m_0} \mathbf{1}_{m_1}^\top$. Let $I_r$ denote the identity matrix of size $r$.  Consequently, its graph Laplacian can be decomposed into:
				\begin{equation}\label{equ_block_L}
					L_H = \begin{pmatrix} 
						L_{H^{[0]}} + b m_1 I_{m_0} & -b \mathbf{1}_{m_0} \mathbf{1}_{m_1}^\top \\ 
						-b \mathbf{1}_{m_1} \mathbf{1}_{m_0}^\top & L_{H^{[1]}} + b m_0 I_{m_1} 
					\end{pmatrix} \enskip.
				\end{equation}
				Write vector $\tilde{u}$ as $\tilde{u} = \big(\sqrt{{m_1} / (m m_0)} \mathbf{1}_{m_0}^\top, - \sqrt{{m_0} / (mm_1)} \mathbf{1}_{m_1}^\top \big)^\top$. Note that $L_{H^{[0]}} \mathbf{1}_{m_0} = \mathbf{0}$ and $L_{H^{[1]}} \mathbf{1}_{m_1} = \mathbf{0}$. Thus, the upper block of $L_H \tilde{u}$ yields:
				\begin{equation*}
					(b m_1 \sqrt{\frac{m_1}{m m_0}} + b m_1 \sqrt{\frac{m_0}{m m_1}}) \mathbf{1}_{m_0}  = m b \sqrt{\frac{m_1}{m m_0}} \mathbf{1}_{m_0} \enskip.
				\end{equation*}
				By symmetry, the lower block computes to $- m b \sqrt{{m_0} / (m m_1)} \mathbf{1}_{m_1}$, establishing  that $L_H \tilde{u} = m b \tilde{u}$.	Thus, $\tilde{u}$ is still an eigenvector of $L_H$ with eigenvalue $m b$.	
				
				Next, we will prove that $u = \tilde{u}$ by induction on number of depth in $H$, denoting as $\mathcal{L}$. When $\mathcal{L} = 0$, all eigenvalues of $L_H$ are $0$.  Thus, $\lambda_1 (L_H) = m b = 0$, and $\tilde{u}$ is naturally an associated eigenvector.
				%When $\mathcal{L} = 0$, according to the last two terms in \eqref{equ 1}, any vector orthogonal to $\mathbf{1}$ is an eigenvector of $L_H$. Thus, the multiplicity of $m b$ is $m -1$. Note that $0$ is an eigenvalue of $L_H$ and $L_H$ has $m$ eigenvalues in total. Thus, $\lambda_1 (L_H) = m b$ and $u = \tilde{u}$.
				
				Assume inductively that claim (1) holds for depth $\mathcal{L} = \cl \geq 1$. We now proceed to the inductive step for $\mathcal{L} = \cl + 1$. 
				Following the root split,  submatrices $H^{[0]}$ and $H^{[1]}$ are themselves HCMs of depth $\cl$. Let $(L_H)^{[0]} \in \mathbb{R}^{m_0 \times m_0}$ and $(L_H)^{[1]} \in \mathbb{R}^{m_1 \times m_1}$ denote the corresponding upper-left and lower-right diagonal blocks of the Laplacian matrix $L_H$ in \eqref{equ_block_L}. Then we have
				\begin{align}
					(L_H)^{[0]} &= L_{H^{[0]}} + m_1 b I_{m_0}, \label{L0}\\
					(L_H)^{[1]} &= L_{H^{[1]}} + m_0 b I_{m_1} \enskip.
				\end{align}
				
				From the inductive hypothesis, the largest eigenvalue of $L_{H^{[0]}}$ is $m_0 b^{0}$. 
				Let $\{v_1,\ldots,v_{m_0}\}$ be an orthonormal eigenbasis of $L_{H^{[0]}}$, where $v_i$ is associated with the eigenvalue $\lambda_i(L_{H^{[0]}})$ for $i=1,\ldots,m_0$, and choose $v_{m_0}=\mathbf1_{m_0} / \sqrt{m_0}$. Consequently,	$\mathbf1_{m_0}^{\top}v_i=0, i=1,\ldots,m_0-1$. 
				For each $i=1,\ldots,m_0-1$, define $\tilde v_i = (v_i^\top, \mathbf0_{m_1}^\top)^\top \in \mathbb{R}^m$. Using \eqref{equ_block_L} and \eqref{L0}, $\tilde v_i$ becomes an eigenvector of $L_H$ with eigenvalue 
				\begin{equation*} \label{lambda}
					\lambda_i(L_H) = \lambda_i(L_{H^{[0]}}) + m_1 b_ \leq m_0 b^{0} + m_1 b < m b \enskip .
				\end{equation*}
				The last inequality holds because $b^{0} < b$ and $m_0 + m_1 = m$. 
				Hence, we know that there are at least $m_0 - 1$ eigenvalues of $L_H$  are smaller than $m b$. 
				Apply the same argument to $L_{H^{[1]}}$ and we obtain another $m_1 - 1$ eigenvalues of $L_H$, all strictly smaller than $m b$. Together with the eigenvalue $0$ associated with $\mathbf1_m$
				and the eigenvalue $mb$ associated with $\tilde u$, these account for all $m$ eigenvalues of $L_H$.
				Therefore, 
				\[
				\lambda_1(L_H)=mb,
				\]
				and hence $u=\tilde u$.
				
				%	From the inductive hypothesis, the largest eigenvalue of $L_{H^{[0]}}$ is $m_0 b^{0}$. Let $v \in \mathbb{R}^{m_0}$ denote the corresponding dominant eigenvector of $L_{H^{[0]}}$. As a non-trivial Laplacian eigenvector, $v$ satisfies $\mathbf{1}_{m_0}^\top v = 0$.  We construct an augmented vector $\tilde{v} \in \mathbb{R}^m$ by zero-padding the remaining $m_1$ coordinates, yielding $\tilde v = ( v^\top, \boldsymbol 0_{m_1}^\top )^\top \in \mathbb{R}^m$. From \eqref{equ_block_L} and \eqref{L0}, $\tilde v$ becomes an eigenvector of $L_H$ with eigenvalue $m_0 b^{0} + m_1 b$.
				%	Similarly, we can construct $m_0$ eigenvalues of $L_H$ as $\lambda_i(L_{H^{[0]}}) + m_1 b$ and they are bounded as follows:
				%	\begin{equation*} \label{lambda}
					%		\lambda_i(L_H) = \lambda_i(L_{H^{[0]}}) + m_1 b_ \leq m_0 b^{0} + %m_1 b < m b \enskip .
					%	\end{equation*}
				%	The last inequality holds because $b^{0} < b$ and $m_0 + m_1 = m$. 
				%	Hence, we know that there are $m_0$ eigenvalues of $L_H$  are smaller than $m b$. 
				%	Apply the same argument to $L_{H^{[1]}}$ and we see that there are $m_1$ eigenvalues of $L_H$  are smaller than $m b$. Since $0$ is an eigenvalue contained in both cases of $L_{H^{[0]}}$ and $L_{H^{[1]}}$, there are $m_0 + m_1 - 1 = m - 1$ eigenvalues of $L_H$ smaller than $m b$. Therefore, we conclude that $mb$ is the largest eigenvalue of $L_H$ and $u = \tilde u$.
				
				Claim (1) has been proved.
				
				\textbf{(ii) Proof of claim (2).}
				
				Building upon the preceding analysis,  the second largest eigenvalue $\lambda_2(L_H)$ is strictly determined by the maximum of the leading eigenvalues induced by $L_{H^{[0]}}$ and $L_{H^{[1]}}$. That is,
				\begin{equation*}
					\lambda_2(L_H) = \max \left\{ m_0 b^0 + m_1 b, \, m_1 b^1 + m_0 b \right\} \enskip.
				\end{equation*}
				
				Claim (2) has been proved.
			\end{proof}
			
			\begin{lemma}\label{lemma eigen range}
				Suppose that the normalized population distance matrix $D \in \mathbb{R}^{m \times m}$ satisfies the HDM structure defined %satisfying the monotonicity property described 
				in Definition \ref{def:HDM}, and that Assumptions \ref{A1}--\ref{A3} hold. Then, we have
				\begin{equation*}
					m \alpha_1 \leq \lambda_1(L_D) \leq m \beta_1.
				\end{equation*}
			\end{lemma}
			
			\begin{proof}
				%		Let $H_\beta$ and $H_\alpha$ be two HCMs sharing the identical block-partitioning structure as $D$ at the first layer, explicitly defined as:
				%		\begin{align*}
					%			H_\beta & = 
					%			\begin{pNiceArray}{c|c}[cell-space-limits = 3pt]
						%				\beta_0 (\mathbf{1}_{m_0} \mathbf{1}_{m_0}^\top - I_{m_0}) & \beta_1 \mathbf{1}_{m_0} \mathbf{1}_{m_1}^\top \\
						%				\hline
						%				\beta_1 \mathbf{1}_{m_1} \mathbf{1}_{m_0}^\top & \beta_0 (\mathbf{1}_{m_1} \mathbf{1}_{m_1}^\top - I_{m_1})
						%			\end{pNiceArray} \enskip, \\		
					%			H_\alpha & = 
					%			\begin{pNiceArray}{c|c}[cell-space-limits = 3pt]
						%				\mathcal{H}_1 & \alpha_1 \mathbf{1}_{m_0} \mathbf{1}_{m_1}^\top \\
						%				\hline
						%				\alpha_1 \mathbf{1}_{m_1} \mathbf{1}_{m_0}^\top & \mathcal{H}_2
						%			\end{pNiceArray} \enskip,
					%		\end{align*}	
				%		where in the diagonal blocks $\mathcal{H}_1$ and $\mathcal{H}_2$, every non-zero entry of the corresponding blocks in $D$ is uniformly replaced by the scalar $\alpha_0$, while all zero entries remain unchanged.
				Building upon the HCM structure formalized in Definition \ref{def:HCM} and illustrated in \eqref{H_layer2}, we construct two HCMs, $H_\beta$ and $H_\alpha$, both modified from $D$. Specifically, $H_\beta$ is constructed by specifying its between-block scalars as $b = \beta_1$ and $b^{[x]} = \beta_1^{[x]}$. The matrix $H_\alpha$ is obtained in the exact same manner by setting $b = \alpha_1$ and $b^{[x]} = \alpha_1^{[x]}$.
				
				%Suppose $v$ is a unit vector. 
				By the property of the Laplacian matrix $L_D$ \citep{von2007tutorial}, we have
				\begin{equation}\label{compute_eigen}
					\lambda_1(L_D) = \max_{\substack{v \in \mathbb{R}^m \\ \|v\|^2 = 1}} v^T L_D v = \max_{\substack{v \in \mathbb{R}^m \\ \|v\|^2 = 1}} \frac{1}{2} \sum_{i, j} D_{ij} (v_i - v_j)^2.
				\end{equation}
				According to the definitions of $H_\alpha$ and $H_\beta$, $H_{\alpha, ij} \leq D_{ij} \leq H_{\beta, ij}$. The value of \eqref{compute_eigen} will not decrease when $D_{ij}$ increases and will not increase when $D_{ij}$ decreases. We obtain
				\begin{equation*} \label{H_beta}
					\lambda_1(L_{H_\alpha}) = \max_{\substack{v \in \mathbb{R}^m \\ \|v\|^2 = 1}} \frac{1}{2} \sum_{i, j} H_{\alpha, ij} (v_i - v_j)^2 \leq \lambda_1(L_D) \leq \max_{\substack{v \in \mathbb{R}^m \\ \|v\|^2 = 1}} \frac{1}{2} \sum_{i, j} H_{\beta, ij} (v_i - v_j)^2 = \lambda_1(L_{H_\beta}).
				\end{equation*}
				Applying Lemma \ref{lemma HCM}, we have $\lambda_1(L_{H_\alpha}) = m \alpha_1, \lambda_1(L_{H_\beta}) = m \beta_1$. Thus,  we have
				\begin{equation*}
					m \alpha_1 \leq \lambda_1(L_D) \leq m \beta_1.
				\end{equation*}
				The proof is completed.
			\end{proof}
			
			\begin{lemma}[Spectrum of HDMs] \label{lemma HDM}
				Let $D \in \mathbb{R}^{m \times m}$ be the normalized population distance matrix defined in \eqref{p-dis}.  Suppose that $D$ is an HDM defined in Definition \ref{def:HDM} and that it satisfies Assumptions \ref{A1}--\ref{A3}.
				Recall the separation $\Delta$ defined as $$\Delta = \alpha_1 - \frac{1}{1 + \tilde \eta} (\beta_0 + \tilde \eta \beta_1) \enskip.$$By Assumptions \ref{asum 4} and \ref{asum 6}, we have $\Delta > 0$ and $\lambda_1(L_D) > \deg_i(D)$ for all $i$. 
				%		Define the separation
				%		\begin{equation}
					%			\Delta = \alpha_1 - \frac{1}{1 + \tilde \eta} (\beta_0 + \tilde \eta \beta_1) \enskip ,
					%		\end{equation}
				%		Assume $\Delta > 0$ and $\lambda_1(L_D) > \deg_i(D)$ for all $i$.  
				Set
				\begin{align*}
					c_1 = \frac{(\alpha_1 - \beta_0)(1 + \tilde\eta)}{(1 + \tilde\eta) \beta_1 - \alpha_1},  \
					c_2 = \frac{\beta_1 (1 + \tilde\eta)}{(1 + \tilde\eta) \alpha_1 - (\beta_0 + \tilde\eta \beta_1)} \enskip.
				\end{align*} 
				Then, 
				\begin{enumerate}[label=(\arabic*)]
					\item The largest and the second largest eigenvalues of $L_D$ satisfy
					\begin{equation*}
						\lambda_1(L_D) - \lambda_2(L_D) \geq  m \Delta.
					\end{equation*}
					\item Let $u$ be the  unit-norm dominant eigenvector of $L_D$, chosen from an orthonormal eigenbasis of $L_D$. Then every entry of $u$ satisfies 
					\begin{equation} \label{v_i}
						\frac{c_1}{\sqrt{m} c_2} \leq |u_i| \leq \frac{c_2}{2 \sqrt m}.
					\end{equation}
				\end{enumerate}
			\end{lemma}

			%			Before we proceed to the proof, we state a result in Spectral Graph Theory that we will use:
			%			\begin{lemma}[Lemma 6.3.1 in \citet{Spielman2009}, Corollary 4.3.12 in \citet{horn2013matrix}] \label{lemma spe}
				%			If $G$ and $H$ are graphs whose Laplacian matrices satisfy 
				%			\[
				%			L_G \succcurlyeq c\,L_H,
				%			\]
				%			for some $c > 0$. That is, $L_G - c\,L_H$ is positive semidefinite, then for each $i$,
				%			\[
				%			\lambda_i(G) \;\ge\; c\,\lambda_i(H).
				%			\]		
				%			\end{lemma}

			\begin{proof}
				\textbf{(i) Proof of claim (1).}
				
				Recalling a property of the graph Laplacian \citep{von2007tutorial}, for any vector $v \in \mathbb{R}^m$, it holds that 
				\begin{equation*}
					v^T L_D v =  \frac{1}{2} \sum_{i, j} D_{ij} (v_i - v_j)^2.
				\end{equation*}
				Generate $H_\beta$ as in the proof of Lemma \ref{lemma eigen range}, we have $L_{H_\beta} \succcurlyeq L_D$. Applying Lemma \ref{lemma spe},
				\begin{equation*}
					\lambda_2(L_D)  \leq \lambda_2(L_{H_\beta}) \enskip.			
				\end{equation*}
				
				Applying claim (2) of Lemma \ref{lemma HCM} to  matrix $H_\beta$, we have the parameters  $b = \beta_1, b^0 = \beta_1^0, b^1 = \beta_1^1$. Furthermore, by incorporating the cluster size ratio $\eta = m_0 / m_1$ and the conditions $b^0 < \beta_0, b^1 < \beta_0$, we obtain the bound of $\lambda_2(L_D)$ as
				\begin{equation*}
					\lambda_2(L_D) \leq \lambda_2(L_{H_\beta}) \leq \max \left\{ \frac{m}{1 + \eta} (\beta_0 + \eta \beta_1), \frac{m}{1 + \eta} (\eta \beta_0 + \beta_1) \right\} = \frac{1}{1 + \tilde \eta} (\beta_0 + \tilde \eta \beta_1) \enskip.
				\end{equation*}
				
				Lemma \ref{lemma eigen range} show that $\lambda_1(L_D) \geq m \alpha_1$. Therefore, we get
				\begin{equation*}
					\lambda_1(L_D) - \lambda_2(L_D) \geq m \alpha_1 - \frac{m}{1 + \tilde \eta} (\beta_0 + \tilde \eta \beta_1)  = m \Delta\enskip.
				\end{equation*}
				Claim (1) has been proved.
				
				\textbf{(ii) Proof of claim (2).}
				
				To establish the bounds on entries of $u$, we consider a single coordinate of $u$. Since  $u$ is the dominant eigenvector of $L_D$, the associated eigenvalue equation directly yields
				\begin{equation*}
					u_i = \frac{D_{i \cdot} u}{\deg_i(D) - \lambda_1(L_D)}.
				\end{equation*}
				Given Assumption \ref{asum 6} that $\lambda_1(L_D) > \deg_i(D)$, we get the absolute value
				\begin{equation}\label{largest u}
					|u_i| = \frac{|D_{i \cdot} u|}{\lambda_1(L_D) - \deg_i(D)}.
				\end{equation}
				Applying Theorem \ref{theorem L_D}, the signs of $u$ perfectly align with the ground-truth cluster structure, meaning $u_i$ is strictly positive  for the first $m_0$ nodes and strictly negative for the remaining $m_1$ nodes. From the fact that $\mathbf{1}$ is an eigenvector of $L_D$, we get that  $u^T \mathbf{1} = 0$  and $\sum_{i: u_i > 0} u_i = \sum_{i: u_i < 0} |u_i| \triangleq J$. Then, we have
				\begin{equation*}
					J^2 = (\sum_{i:u_i > 0} u_i)^2 \leq m_0 \sum_{i:u_i > 0}u_i^2 = m_0 \big(1 - \sum_{i:u_i < 0} u_i^2\big) \leq m_0 \Big(1 - \frac{1}{m_1} \big(\sum_{i:u_i < 0} |u_i|\big)^2\Big) = m_0 - \frac{m_0}{m_1} J^2 \enskip.
				\end{equation*}
				The two inequality hold true by the Cauchy-Schwarz inequality. Hence,
				\begin{equation} \label{J 1}
					J \leq \sqrt{\frac{m_0 m_1}{m}} = \sqrt{\frac{m_0(m - m_0)}{m}} \leq \frac{\sqrt m}{2}.
				\end{equation}
				The numerator of \eqref{largest u} is bounded as:
				\begin{equation}\label{numerator}
					J(\alpha_1 - \beta_0) \leq |D_{i \cdot} u| \leq J \beta_1 \enskip.
				\end{equation}
				Next, we will derive the bound on the denominator of \eqref{largest u}. Obviously,
				\begin{align} \label{di}
					\frac{m}{1 + \eta} \min\{ \alpha_1, \eta \alpha_1\} & \leq \deg_i(D) \leq \frac{m}{1 + \eta} \max\{\beta_0 + \eta \beta_1, \eta\beta_0 +  \beta_1\} \nonumber\\
					\frac{m}{1 + \tilde \eta} \alpha_1 & \leq \deg_i(D)  \leq \frac{m}{1 + \tilde \eta} (\beta_0 + \tilde \eta \beta_1) \enskip.
				\end{align}
				Combining Lemma \ref{lemma eigen range} and \eqref{di}, the denominator can be bounded as
				\begin{equation}\label{denominator}
					m \alpha_1  - \frac{m}{1 + \tilde \eta} (\beta_0 + \tilde \eta \beta_1) < \lambda_1(L_D) - \deg_i(D) < m \beta_1 - \frac{m}{1 + \tilde \eta} \alpha_1 \enskip.
				\end{equation}
				From \eqref{numerator} and \eqref{denominator}, we obtain
				\begin{align} \label{ui}
					\frac{J c_1}{m} \leq |u_i| \leq \frac{J c_2}{m} \enskip.
				\end{align}
				Since $u$ is a unit vector, we can bound $J$ as
				\begin{align}
					m \left(\frac{J c_2}{m}\right)^2 & \geq 1 \nonumber \\
					J & \geq \frac{\sqrt m}{c_2} \enskip. \label{J 2}
				\end{align}
				Applying \eqref{J 1}, \eqref{ui},  and \eqref{J 2}, we get that 
				\begin{equation*}
					\frac{c_1}{ \sqrt{m} c_2} \leq |u_i| \leq \frac{c_2}{2 \sqrt m} \enskip.
				\end{equation*}
				
				Claim (2) has been proved.
			\end{proof}
			
			%			\textbf{Graphon estimation error.}
			%			Under the assumption that an NBS-type estimator is used for each link probability matrix, the estimation error satisfies the following bound \citep{zhao2019change}.
			%			\begin{equation}\label{NBS}
				%			\|\hat P^{(i)} - P^{(i)} \|_F \leq C_0 \sqrt{n \log n} \quad \text{for all} \ i = 1, 2, \ldots, m,
				%			\end{equation}
			%			with probability at least $1 - n^{- \varepsilon}$ for any $\varepsilon > 0$, where $C_0$ is a positive global constant depending on $\varepsilon$ and another global constant $B_0 > 0$ but not on $n$ or $m$.
			
			%	\begin{lemma}[Davis-Kahan for the dominant eigenvector] \label{lemma DK}
				%		Suppose that $\{\hat P^{(i)}\}_{i = 1}^m$ are the estimated link probability matrices of $\{P^{(i)}\}_{i = 1}^m$ applying NBS method from \citet{zhang2017estimating}. 
				%		Let $D$ be the true normalized population distance matrix and $\hat D$ an estimated distance matrix built from $\{\hat P^{(i)}\}_{i = 1}^m$. 
				%		Let $u$ and $\hat u$ be dominant eigenvectors of $L_D$ and $L_{\hat D}$, respectively, both with unit norm.  Then, under the same assumptions as Lemma \ref{lemma HDM}, for any constant $\varepsilon > 0$,  with probability at least $(1 - n^{- \varepsilon})^m$, we have
				%		\begin{equation*}
					%			\|\hat u - u\| \leq \frac{C_1 \sqrt{\mathstrut n \log n}}{\Delta n}  \enskip.
					%		\end{equation*}	
				%		for some constant $C_1 > 0$.
				%	\end{lemma} 
			\begin{lemma}[Davis-Kahan for the dominant eigenvector] \label{lemma DK}
				Suppose that $\{\hat P^{(i)}\}_{i = 1}^m$ are the estimated link probability matrices of $\{P^{(i)}\}_{i = 1}^m$ satisfying Assumption \ref{asum 7}. 
				Let $D$ be the true normalized population distance matrix and $\hat D$ an estimated distance matrix built from $\{\hat P^{(i)}\}_{i = 1}^m$. 
				Let $u$ and $\hat u$ be dominant eigenvectors of $L_D$ and $L_{\hat D}$, respectively, both with unit norm.  Then, under the same assumptions as Lemma \ref{lemma HDM}, for any constant $\varepsilon > 0$,  with probability at least $(1 - n^{- \varepsilon})^m$, we have
				%	\begin{equation*}
					%		\|\hat u - u\| \leq \frac{C_1 \sqrt{\mathstrut n \log n}}{\Delta n}  \enskip.
					%	\end{equation*}	
				%	for some constant $C_1 > 0$.
				\begin{equation*}
					\|\hat u - u\| \leq \frac{2^{7/2} \zeta}{\Delta}  \enskip,
				\end{equation*}	
				where $\|\cdot\|$ denotes the Euclidean norm.
			\end{lemma} 
			
			\begin{proof}
				The result follows from a slight variant of Davis-Kahan theorem that appears in \citet{yu2015useful}. Applying Theorem 2 of \citet{yu2015useful} with $r = s = 1$, we have 
				\begin{equation*}
					\|\hat u - u\| \leq \frac{2^{3/2} \|L_D - L_{\hat D}\|_{\mathrm{op}}}{\lambda_1(L_D) - \lambda_2(L_D)} \enskip.
				\end{equation*}
				
				Next, we derive the bound of $\| L_D - L_{\hat D} \|_{\mathrm{op}}$. By applying the triangle inequality, we have
				\begin{align*}
					\| L_D - L_{\hat D} \|_{\mathrm{op}}  &= \| \mathrm{diag}(D\mathbf 1)  - D - (  \mathrm{diag}(\hat D\mathbf 1) - \hat D) \|_{\mathrm{op}} \\
					&\leq  \| \mathrm{diag}((D - \hat D)\mathbf 1) \|_{\mathrm{op}} + \| D - \hat D \|_{\mathrm{op}} \\
					&\leq \max_i \left| \sum_{j = 1}^m \left(D_{ij} - \hat D_{ij}\right) \right| + \| D - \hat D \|_F \enskip.
				\end{align*}
				The last inequality follows from the fact that the operator of  a diagonal matrix is the largest absolute  element on its diagonal,  and the norm relaxation $\| \cdot \|_{\mathrm{op}} \leq \| \cdot \|_F$. 
				To control the entry-wise estimation error, we observe that
				\begin{align*} 
					| D_{ij} - \hat D_{ij}| & = \frac{1}{n} \left|\|P^{(i)} - P^{(j)}\|_F  - \|\hat P^{(i)} - \hat P^{(j)}\|_F \right| \nonumber  \\
					& \leq \frac{1}{n} \| (P^{(i)} - P^{(j)}) - (\hat P^{(i)} - \hat P^{(j)})\|_F \\
					& \leq \frac{1}{n}  \|P^{(i)} - \hat P^{(i)}\|_F + \frac{1}{n} \| P^{(j)} - \hat P^{(j)}\|_F \enskip .\nonumber
				\end{align*}
				%		Using \eqref{NBS}, we obtain that for all $i$ and $j$,
				According to Assumption \ref{asum 7}, we obtain that for all $i$ and $j$,
				\begin{equation} \label{D-D_hat-element}
					| D_{ij} - \hat D_{ij}| \leq 2 \zeta \enskip ,%2 C_0 \frac{ \sqrt{\mathstrut n \log n}}{n} \enskip ,
				\end{equation}
				with probability at least $(1 - n^{- \varepsilon})^m$.
				Therefore, 				
				\begin{equation} \label{D-D_hat}
					\| L_D - L_{\hat D} \|_{\mathrm{op}}  \leq 4 m \zeta \enskip, %4 C_0 m \frac{ \sqrt{\mathstrut n \log n}}{n} \enskip ,
				\end{equation}
				with probability at least $(1 - n^{- \varepsilon})^m$.
				
				According to conclusion (1) of Lemma \ref{lemma HDM}, it follows that with probability at least $(1 - n^{- \varepsilon})^m$,
				\begin{equation*}
					\|\hat u - u\| \leq \frac{2^{3/2} \cdot 4 m \zeta}{ m \Delta} = \frac{2^{7/2} \zeta}{\Delta}   \enskip. %\frac{2^{3/2} \cdot 4 C_0 m \sqrt{n \log n}}{ m \Delta n} = \frac{C_1 \sqrt{\mathstrut n \log n}}{\Delta n}   \enskip,
				\end{equation*}	
				%		with $C_1 = 2^{7/2} C_0$.
			\end{proof}
			
			In the empirical setting, Lemma \ref{lemma DK} provides an upper bound on the eigenvector deviation $\|\hat u - u\|$. To demonstrate that the spectral clustering can make no mistakes with high probability, we need to verify that $\hat u$ is uniformly close to $u$ in every coordinate ensuring their sign patterns coincide.
			Without loss of generality, we focus on an arbitrary cluster with binary label $x$. To avoid notational clutter, we suppress the index $x$, letting $m$ and $D$ denote its size and the corresponding population distance matrix. Following Lemma \ref{lemma DK}, let $u$ and $\hat{u}$ denote the dominant eigenvectors of the true Laplacian $L_D$ and its empirical counterpart $L_{\hat{D}}$, respectively, both with unit norm. Repeated application for all of the clusters will guarantee the correctness of our algorithm across all clusters.

			\begin{proof}[of Theorem~\ref{theorem sign}]	
				We establish the theorem by sequentially proving \eqref{sign} and the convergence result.
				
				\textbf{(i) Proof of result \eqref{sign}.}
				%	\subsection{Proof of result \eqref{sign}}
				
				Let $e = \hat u - u$. The key point is to show that with high probability  the element-wise perturbation $e_i$ is uniformly low for all $i = 1, \ldots, m$.	Denote the largest eigenvalues of $L_D$ and $L_{\hat D}$ as $\lambda_1$ and $\hat \lambda_1$ respectively. From the definition of eigenvector we have
				\begin{equation*}
					(L_{\hat D})_{i \cdot} \hat u = \hat \lambda_1 \hat u_i = \hat \lambda_1 (e_i + u_i) \enskip.
				\end{equation*}
				Thus,
				\begin{align*}
					\hat \lambda_1 e_i & = (L_{\hat D})_{i \cdot} \hat u - \hat \lambda_1 u_i \\
					\quad & = \big((L_D)_{i \cdot} + (L_{\hat D - D})_{i \cdot}\big) (e + u) - \hat \lambda_1 u_i\nonumber\\
					\quad & = \deg_i(D) e_i - D_{i \cdot} e + \lambda_1 u_i + \deg_i(\hat D - D) \big(e_i + u_i\big) - (\hat D - D)_{i \cdot} (e + u) - \hat \lambda_1 u_i \enskip. \nonumber
				\end{align*}
				Therefore,
				\begin{equation*}
					e_i = \frac{\overbrace{\big(\lambda_1 - \hat \lambda_1 + \deg_i(\hat D - D)\big) u_i}^{\mathcal{T}_1}  - \overbrace{ D_{i \cdot} e}^{\mathcal{T}_2} - \overbrace{ (\hat D - D)_{i \cdot} (e + u)}^{\mathcal{T}_3}}{\underbrace{\hat \lambda_1 - \deg_i(\hat D)}_{\mathcal{T}_4}} \enskip.
				\end{equation*}
				Call the numerator terms $\mathcal{T}_1, \mathcal{T}_2$ and $\mathcal{T}_3$, and the denominator term $\mathcal{T}_4$. Our goal is to bound $|e_i|$ uniformly in $i$ by controlling $\mathcal{T}_1, \mathcal{T}_2, \mathcal{T}_3$ and $\mathcal{T}_4$.
				
				\textbf{Bound on $\mathcal{T}_1$:}
				
				First, applying Lemma \ref{weyl}, 
				\begin{equation*}
					|\lambda_1 - \hat{\lambda}_1| \leq \| L_D - L_{\hat D} \|_{\mathrm{op}} \enskip.
				\end{equation*}
				According to \eqref{D-D_hat}, we obtain 
				\begin{equation}\label{ineq1}
					|\lambda_1 - \hat{\lambda}_1| \leq 4 m \zeta \enskip, %4 C_0 \frac{m  \sqrt{\mathstrut n \log n}}{n} \enskip,
				\end{equation}
				with probability at least $(1 - n^{- \varepsilon})^m$. 		
				
				Second, applying \eqref{D-D_hat-element}, $|\deg_i(\hat D - D)|$ can be bounded as
				\begin{align}\label{ineq2}
					|\deg_i(\hat D - D)|  \leq \max_{i = 1, 2, \ldots, m}|\sum_{j = 1}^m (D_{ij} - \hat{D}_{ij})| 
					\quad  \leq \max_{i = 1, 2, \ldots, m}\sum_{j = 1}^m |D_{ij} - \hat{D}_{ij}| 
					\quad  \leq 2 m \zeta \enskip,  %2 C_0 \frac{m  \sqrt{\mathstrut n \log n}}{n} \enskip, \nonumber
				\end{align}
				with probability at least $(1 - n^{- \varepsilon})^m$.
				
				From \eqref{ineq1} and \eqref{ineq2}, with probability at least $(1 - n^{- \varepsilon})^m$, we have
				\begin{equation} \label{T1 part 1}
					|\lambda_1 - \hat \lambda_1 + \deg_i(\hat D - D)| \leq 6 m \zeta \enskip. %6 C_0 \frac{m  \sqrt{\mathstrut n \log n}}{n} \enskip.
				\end{equation} 
				
				Last, combining with \eqref{T1 part 1} and the bound of $|u_i|$ in Lemma \ref{lemma HDM}, we have
				\begin{equation*}
					|\mathcal{T}_1| \leq \frac{c_2}{2 \sqrt{m}} 6 m \zeta \triangleq C_{\mathcal{T}_1} \sqrt{m} \zeta \enskip, %\leq \frac{c_2}{2 \sqrt{m}} 6 C_0 \frac{m  \sqrt{\mathstrut n \log n}}{n} \triangleq C_{\mathcal{T}_1} \frac{ \sqrt{m \mathstrut n \log n}}{n} \enskip,
				\end{equation*}
				where $ C_{\mathcal{T}_1} = 3 c_2$,
				%		$ C_{\mathcal{T}_1} = 3 C_0 c_2$, 
				holding with probability at least $(1 - n^{- \varepsilon})^m$.
				
				\textbf{Bound on $\mathcal{T}_2$:}
				
				By the upper bound $D_{ij} \le \beta_1$ and Assumption \ref{asum 5}, it follows that
				\begin{equation*}
					\|D_{i \cdot}\| = \big(\sum_j D_{ij}^2\big)^{1/2} \leq \sqrt{m} \beta_1  = \sqrt{m} O(\Delta) \enskip.
				\end{equation*}
				Applying the error bound in Lemma \ref{lemma DK}, we obtain, with probability at least $(1 - n^{- \varepsilon})^m$, that
				\begin{equation*}
					|\mathcal{T}_2| \leq \|D_{i \cdot}\| \|e\| = \|D_{i \cdot}\| \|\hat u - u\|  \leq  \frac{2^{7/2} \sqrt{m} \zeta}{\Delta} O(\Delta) \enskip. %\leq \sqrt{m} O(\Delta) \cdot \frac{C_1 \sqrt{\mathstrut n \log n}}{\Delta n} = \frac{C_1 \sqrt{m \mathstrut n \log n}}{\Delta n} O(\Delta) \enskip.
				\end{equation*}
				For sufficient large $n$, there exists some constant $C_1 > 0$ such that with probability at least $(1 - n^{-\varepsilon})^{m}$,
				\begin{equation*}
					|\mathcal{T}_2| \leq 2^{7/2} C_1 \sqrt{m} \zeta
					\triangleq C_{\mathcal{T}_2} \sqrt{m} \zeta \enskip . % C_1 C_2  \frac{ \sqrt{m \mathstrut n \log n}}{n}
					%			\triangleq C_{\mathcal{T}_2} \frac{ \sqrt{m \mathstrut n \log n}}{n} \enskip .
				\end{equation*}
				with $C_{\mathcal{T}_2} = 2^{7/2}  C_1$.
				%		with $C_{\mathcal{T}_2} = C_1 C_2$.
				
				\textbf{Bound on $\mathcal{T}_3$:}
				
				Using Lemma \ref{lemma HDM} and Lemma \ref{lemma DK}, we bound $\|\hat u\|$ as follows
				%		\begin{equation*}
					%			\|\hat u\| \leq \|e\| + \|u\| \leq \frac{ C_1 \sqrt{\mathstrut n \log n}}{\Delta n} + \frac{c_2}{2 \sqrt{m}} \sqrt{m} = \frac{ C_1 \sqrt{\mathstrut n \log n}}{\Delta n} + \frac{c_2}{2} \enskip,
					%		\end{equation*}
				\begin{equation*}
					\|\hat u\| \leq \|e\| + \|u\| \leq \frac{ 2^{7/2} \zeta}{\Delta} + \frac{c_2}{2 \sqrt{m}} \sqrt{m} = \frac{ 2^{7/2} \zeta}{\Delta} + \frac{c_2}{2} \enskip,
				\end{equation*}
				with probability at least $(1 - n^{- \varepsilon})^m$.
				Since $\Delta = \omega(\zeta)$, when $n$ is large enough, there exists constant $C_2 > 0$ such that 
				\[\|\hat u\| \leq C_2.\]
				Applying \eqref{D-D_hat-element}, we obtain that with probability at least $(1 - n^{-\varepsilon})^{m}$,
				%		\begin{equation*}
					%			\|(\hat D - D)_{i \cdot}\| = \Big(\sum_{j = 1}^m (\hat D_{ij} - D_{ij})^2\Big)^{1/2} \leq \Big(\big(2 C_0 \frac{ \sqrt{\mathstrut n \log n}}{n}\big)^2 m \Big)^{1 / 2} = 2 C_0 \frac{ \sqrt{m \mathstrut n \log n}}{n} \enskip.
					%		\end{equation*}
				\begin{equation*}
					\|(\hat D - D)_{i \cdot}\| = \Big(\sum_{j = 1}^m (\hat D_{ij} - D_{ij})^2\Big)^{1/2} \leq \Big(\big(2 \zeta)^2 m \Big)^{1 / 2} = 2 \sqrt{m} \zeta \enskip.
				\end{equation*}
				Therefore, with probability at least $(1 - n^{-\varepsilon})^{m}$,
				%		\begin{equation*}
					%			|\mathcal{T}_3| \leq \|(\hat D - D)_{i \cdot}\| \|\hat u\| \leq 2 C_0 C_3  \frac{ \sqrt{m \mathstrut n \log n}}{n}  \triangleq C_{\mathcal{T}_3} \frac{ \sqrt{m \mathstrut n \log n}}{n}  \enskip,
					%		\end{equation*}
				\begin{equation*}
					|\mathcal{T}_3| \leq \|(\hat D - D)_{i \cdot}\| \|\hat u\| \leq 2 C_2 \sqrt{m} \zeta  \triangleq C_{\mathcal{T}_3} \sqrt{m} \zeta  \enskip,
				\end{equation*}
				with $C_{\mathcal{T}_3} = 2 C_2$.
				%		with $C_{\mathcal{T}_3} = 2 C_0 C_3$.
				
				\textbf{Bound on $\mathcal{T}_4$:} 
				
				First, by Lemma \ref{lemma eigen range} and Assumptions \ref{A1}--\ref{A3},  
				\begin{equation*}
					\lambda_1 - \deg_i(D)  \geq m \alpha_1 - \max\{m_0 \beta_0 + m_1 \beta_1, m_0 \beta_1 + m_1 \beta_0\}
					= m \alpha_1 - \frac{m}{1 + \tilde{\eta}} (\beta_0 + \tilde{\eta} \beta_1) 
					= m \Delta \enskip.
				\end{equation*}
				Utilizing the triangle inequality, we have that with probability at least $(1 - n^{-\varepsilon})^{m}$,
				\begin{align*}
					|\hat \lambda_1 - \deg_i(\hat D)| & = |\lambda_1 - \deg_i(D) + \big(\hat \lambda_1 - \lambda_1 + \deg_i(D - \hat D)\big)| \nonumber\\
					\quad & \geq (\lambda_1 - \deg_i(D)) - |\hat \lambda_1 - \lambda_1 + \deg_i(D - \hat D)| \nonumber\\
					\quad & \geq  m \Delta - 6 m \zeta \enskip .
					%\quad & = \omega(m \zeta) \enskip . 
					%			\quad & \geq  m \Delta - 6 C_0 \frac{m  \sqrt{\mathstrut n \log n}}{n} \nonumber\\
					%			\quad & = \omega( \frac{m  \sqrt{\mathstrut n \log n}}{n} ) \enskip . 
				\end{align*}
				The last inequality comes from Assumption \ref{asum 6} and inequality \eqref{T1 part 1}. % and the last equality comes from the Assumption \ref{asum 4}.
				
				Combining all the bound results, we have that with probability at least $(1 - n^{-\varepsilon})^{m}$,
				\begin{equation*}
					|e_i| = \left|\frac{\mathcal{T}_1 - \mathcal{T}_2 - \mathcal{T}_3}{\mathcal{T}_4}\right| \leq  \frac{(C_{\mathcal{T}_3} + C_{\mathcal{T}_3} + C_{\mathcal{T}_3}) \sqrt{m} \zeta}{m \Delta - 6 m \zeta} 
					\triangleq  \frac{1}{\sqrt{m} } \cdot \frac{C}{ \Delta/\zeta - 6}  \enskip,
				\end{equation*}
				where $C = C_{\mathcal{T}_3} + C_{\mathcal{T}_3} + C_{\mathcal{T}_3}$. To examine the asymptotic behavior, consider any fixed $\varepsilon > 0$. Under the asymptotic regime where $n \to \infty$ and Assumption \ref{asum 4}, $\Delta = \omega(\zeta)$, implying that 
				\begin{equation*}
					|e_i| \leq  \frac{1}{\sqrt{m}} \cdot O\left( \frac{1}{\Delta/\zeta} \right) = o(\frac{1}{\sqrt{m}}) \enskip.
				\end{equation*}
				
				On the other hand, by the lower bound in \eqref{v_i}, $|u_i| \geq {c_1} / (c_2 \sqrt{m})$. Therefore, for sufficiently large $n$, we have
				\begin{equation*}
					|e_i| < \frac{c_1}{ c_2 \sqrt{m}} \leq  |u_i|.
				\end{equation*}
				This implies that $\text{sign}(\hat u) = \text{sign}(u)$.
				The bound conclusion \eqref{sign} holds true with probability at least $(1 - n^{-\varepsilon})^{m}$ for any $\varepsilon > 0$.  
				
				This completes the proof of \eqref{sign}.	
				
				\textbf{(ii) Proof of convergence result.}
				Following Bernoulli's inequality, which states that for any integer $\mathcal{m}\ge1$ and real number $y\ge-1$,	$(1+y)^{\mathcal{m}}\ge 1+\mathcal{m}y$, we obtain
				\[
				(1-n^{-\varepsilon})^m \ge 1-mn^{-\varepsilon}.
				\]
				Since $mn^{-\varepsilon}=o(1)$ as $n\to\infty$, it follows that
				\[
				(1-n^{-\varepsilon})^m \to 1.
				\]
				
				%\begin{equation*}
				%	\lim_{\substack{n \to \infty \\ m \to \infty}} (1 - n^{-\varepsilon})^{m} = \lim_{\substack{n \to \infty \\ m \to \infty}} (1 +  \frac{1}{- n^{\varepsilon}} )^{- n^{\varepsilon} \cdot (- n^{-\varepsilon})m} 
				%	= \lim_{\substack{n \to \infty \\ m \to \infty}} e^{- m n^{-\varepsilon}} 
				%	= \lim_{\substack{n \to \infty \\ m \to \infty}} e^{- \frac{m}{n^{\varepsilon}}}  
				%	= 1 \enskip .
				%	\end{equation*}
			%	The last equality follow the assumption $\frac{m}{n^\varepsilon} = o(1)$. 
			Therefore, Method $\text{NHC}_{\text{L}}$ asymptotically recovers the sign pattern, i.e., $\text{sign}(\hat u) = \text{sign}(u)$ with probability tending to 1.
		\end{proof}
		
		\section{Proof of Theorem \ref{theorem consistency}}
		Let $E_\Gamma$ be the event that $\Gamma$ is exactly recovered by method $\text{NHC}_{\text{L}}\text{-TST}$. To establish the lower bound for $P(E_\Gamma)$, we decompose the global failure into local decision errors at each cluster with binary label $x$.
		
		To establish the global recovery guarantee, we bound the total probability of failure by applying a union bound over all nodes in the true hierarchical structure. Crucially, to circumvent the intractability of error propagation, we evaluate the local error probability at each node conditional on the event that all its ancestral splits were perfectly recovered.
		
		First, consider the spectral clustering step at any internal cluster $x$ where $\delta_x = 1$. The node represents a sub-group of size $m_x$, naturally yielding an $m_x \times m_x$ sub-distance matrix.	
		Let  $E_s(x)$ denote the event that one spectral bipartition (i.e., one hierarchical split) for cluster $x$ fails to recover the exact structure. 
		From Theorem \ref{theorem sign}, the success probability is bounded by
		\begin{equation*}
			P(\bar E_s(x)) \geq (1 - n^{- \varepsilon})^{m_x} \geq 1 - m_x n^{- \varepsilon}.
		\end{equation*}
		%The final inequality follows directly from Bernoulli's inequality, which states that for any integer $\mathcal{m} \geq 1$ and real number $y \geq - 1$, we have $(1 + y)^\mathcal{m} \geq 1 + \mathcal{m} y$.	
		Thus, the probability of a clustering failure at node $x$ is upper-bounded by
		\begin{equation} \label{P_E}
			P( E_s(x)) \leq m_x n^{- \varepsilon}.
		\end{equation}	
		
		Second, $\text{NHC}_{\text{L}}\text{-TST}$ relies on the two-sample testing approach adopted  to serve as a statistical stopping rule. 
		Let $E_{I}(x)$ denote the event of Type I error for a terminal leaf node $x$ with $\delta_x = 0$ and $E_{II}(x)$ denote the event of Type II error for an internal node $x$ with $\delta_x = 1$. 
		A global exact recovery event $E_\Gamma$ requires that no bipartition errors and no testing errors occur across the entire tree. Considering all events up to depth $\mathcal{L}$, and applying the union bound along with \eqref{P_E}, the probability of global success is:	
		\begin{align*}
			P(E_\Gamma) & = P \biggl( \bigcap_{l = 0}^{\mathcal{L} - 1} \bigcap_{\substack{x \in \{0,1\}^l \\ \delta_x = 1}} \bar E_s(x)  \bar E_{II}(x)   \bigcap_{l = 0}^{\mathcal{L}} \bigcap_{\substack{x \in \{0,1\}^l \\ \delta_x = 0}}  \bar E_{I}(x) \biggr) \nonumber \\
			& =  1 -P \biggl( \bigcup_{l = 0}^{\mathcal{L} - 1} \bigcup_{\substack{x \in \{0,1\}^l \\ \delta_x = 1}} \big( E_s(x) \bigcup E_{II}(x) \big) \bigcup_{l = 0}^{\mathcal{L}} \bigcup_{\substack{x \in \{0,1\}^l \\ \delta_x = 0}}  E_{I}(x) \biggr)\nonumber \\
			& \geq 1 - \sum_{l = 0}^{\mathcal{L} - 1} \sum_{\substack{x \in \{0,1\}^l \\ \delta_x = 1}} P\big( E_s(x)\big) 
			- \sum_{l = 0}^{\mathcal{L} - 1} \sum_{\substack{x \in \{0,1\}^l \\ \delta_x = 1}} P\big( E_{II}(x)\big)
			- \sum_{l = 0}^{\mathcal{L}} \sum_{\substack{x \in \{0,1\}^l \\ \delta_x = 0}} P\big( E_{I}(x)\big) \nonumber \\
			& \geq 1 - \sum_{l = 0}^{\mathcal{L} - 1} \sum_{\substack{x \in \{0,1\}^l \\ \delta_x = 1}} (m_x n^{- \varepsilon} +  \beta_x)
			- 2^{\mathcal{L}}  \alpha_n   \enskip.
		\end{align*}	
		
		Finally, we evaluate the asymptotic behavior as $n \to \infty$. By the assumption $m_x = o(n^\varepsilon)$, the clustering error $m_x n^{-\varepsilon} \to 0$. Provided that the employed two-sample test is asymptotically powerful, the Type II error rate satisfies $\beta_x = o(1)$. Meanwhile, $\alpha_n = o(1)$ guarantees that the last term goes to 0 as $n \to \infty$.
		
		Therefore, $\lim\limits_{\substack{n \to \infty \\ m_x = o(n^\varepsilon)}} P(E_T) = 1$, which completes the proof.

		%Finally, we evaluate the asymptotic behavior as $n \to \infty$. By the assumption $m_x = o(n^\varepsilon)$, the clustering error $m_x n^{-\varepsilon} \to 0$. 
		%Furthermore, Theorem 3 of \citet{chen2024spectral} guarantees an asymptotic power of 1 for the two-sample test, provided certain theoretical conditions are met. Given that \citet{chen2024spectral} separately prove the MNBS procedure satisfies these exact requirements, our use of MNBS (Algorithm \ref{NHC-TST}, lines \ref{alg:est_2}) inherently secures an asymptotic power of 1.	
		%Consequently, the test is asymptotically powerful, ensuring the Type II error rate $\beta_x = o(1)$. Meanwhile, $\alpha = o(2^{-\cL})$ guarantees that the last term goes to 0 as $n \to \infty$.
		
		%Therefore, $\lim\limits_{\substack{n \to \infty \\ m_x = o(n^\varepsilon)}} P(E_T) = 1$, which completes the proof.
		%\end{proof}
		
	\end{document}